\documentclass[aps,reprint,twocolumn,10pt,superscriptaddress]{revtex4-2}

\usepackage{mathrsfs}
\usepackage{bbm}
\usepackage{amsmath}
\usepackage{amssymb}
\usepackage{amsthm}
\usepackage{graphicx}
\usepackage{MnSymbol}
\usepackage{mathtools}
\usepackage{physics}
\usepackage{dsfont}

\makeatletter
\usepackage{hyperref}
\usepackage{color}
\definecolor{supcol}{RGB}{10,50,180}
\definecolor{eqcol}{RGB}{220,10,100}
\hypersetup{
	colorlinks,
	citecolor=supcol,
	linkcolor=eqcol,
	urlcolor=supcol
}

\allowdisplaybreaks

\newtheorem{theorem}{Theorem}

\newtheorem{proposition}[theorem]{Proposition}

\newtheorem{lemma}[theorem]{Lemma}

\DeclareMathOperator{\diag}{diag}
\DeclareMathOperator{\atanh}{atanh}

\DeclareMathOperator{\argmax}{argmax}
\DeclareMathOperator{\argmin}{argmin}

\newcommand{\mca}{\mathcal}
\newcommand{\mbb}{\mathbb}

\newcommand{\msf}{\mathsf}
\newcommand{\mfr}{\mathfrak}
\newcommand{\mds}{\mathds}

\newcommand{\mbm}[1]{\boldsymbol{#1}}

\newcommand{\Id}{\msf{Id}}

\begin{document}
\title{Geometric Complexity in Thermodynamics}

\author{Tan Van Vu}
\email{tan.vu@yukawa.kyoto-u.ac.jp}
\affiliation{Center for Gravitational Physics and Quantum Information, Yukawa Institute for Theoretical Physics, Kyoto University, Kitashirakawa Oiwakecho, Sakyo-ku, Kyoto 606-8502, Japan}

\author{Keiji Saito}
\email{keiji.saitoh@scphys.kyoto-u.ac.jp}
\affiliation{Department of Physics, Kyoto University, Kyoto 606-8502, Japan}

\date{\today}

\begin{abstract}
The third law of thermodynamics forbids cooling a physical system to absolute zero in a finite number of operational steps. Although this unattainability principle has been quantified for specific state-to-state transitions, a universal, dynamics-independent bound for implementing a state-agnostic reset map remains elusive. In this work, we unveil the fundamental limits of physical map implementation by deriving a trade-off relation based on geometric complexity. By analyzing continuous paths of maps on a geometric manifold, we prove that the geometric complexity of any classical stochastic map or quantum channel is bounded from below by its execution error. As a consequence, we show that achieving zero error in a state-reset operation requires a divergent geometric complexity---a unified measure that naturally incorporates disparate physical resources, including infinite time, energetic cost, or control bandwidth. This unattainability principle holds universally across both classical and quantum regimes, establishing a strict geometric limit on the physical realization of reset operations in thermodynamic control and quantum computation.
\end{abstract}

\maketitle

The third law of thermodynamics places a definitive boundary on the manipulation of physical systems: according to the Nernst statement \cite{Nernst.1906}, ``{\it it is impossible for any process, no matter how idealized, to reduce the entropy of a system to its absolute-zero value in a finite number of operations}.'' In the context of control and computation, this principle translates into a fundamental prohibition against the perfect initialization, or reset, of a system to a pure state \cite{Franco.2013.SR,Wu.2013.SR,Ticozzi.2014.SR,Freitas.2018}. Because resetting a state inherently requires the extraction of entropy, it is not merely an algorithmic prerequisite for classical and quantum computing, but a thermodynamic process governed by strict physical limits \cite{Landauer.1961.JRD,Brut.2012.N,Goold.2015.PRL,Proesmans.2020.PRL,Zhen.2021.PRL,Soldati.2022.PRL,Vu.2022.PRL,Lee.2022.PRL,Scandi.2022.PRL,Vu.2023.PRX,Rolandi.2023.PRL,Munson.2025.PRXQ,Taranto.2025.PRL,QTRoadmap.2026.QST}.

While the third law is traditionally framed in terms of execution time, recent studies have revealed that time is not the sole decisive factor for microscopic systems. Other physical resources---such as thermokinetic cost, control bandwidth, and the size of heat baths---also play a crucial role \cite{Reeb.2014.NJP,Taranto.2023.PRXQ,Masanes.2017.NC,Vu.2025.PRX}. Accordingly, quantitative relations for the third law have been derived to characterize the trade-off between these finite resources and the achievable error \cite{Masanes.2017.NC,Wilming.2017.PRX,Scharlau.2018.Q,Clivaz.2019.PRL,Vu.2025.PRX,Allahverdyan.2011.PRE,Levy.2012.PRL}. However, because the extant studies focus on specific, state-to-state dynamical processes where relevant resources are narrowly defined, the resulting bounds are highly dynamics-dependent. A change in the underlying dynamics or control protocol can readily render them inapplicable.

Crucially, a true reset operation cannot rely on bespoke protocols tailored to specific inputs; it must be implemented as a state-agnostic map capable of resetting an arbitrary, potentially unknown initial state. Such maps arise not only in cooling but more generally whenever probability must be removed from a particular subspace, as in state erasure or copying \cite{Vu.2025.PRX}, and they can be physically realized through a wide variety of Markovian and non-Markovian pathways. Consequently, the difficulty of approaching a perfect reset manifests differently depending on the implementation---perhaps as a divergent execution time, the demand for infinite wall potentials, or the necessity of infinitely fast transition rates. This fragmentation gives rise to a fundamental question: Can the sheer difficulty of physically implementing reset maps be captured by a single, unified principle?

\begin{figure}[b]
\centering
\includegraphics[width=1\linewidth]{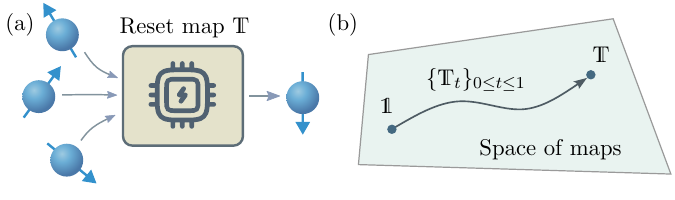}
\protect\caption{{\bf Geometric framework for reset maps.} (a) A state-agnostic reset map $\mds{T}$ clearing probability from a target subspace. This unified description encompasses diverse thermodynamic operations, including cooling, information erasure, and state copying. (b) The geometric complexity of $\mds{T}$ is defined as the minimal geodesic length between the identity map $\mds{1}$ and $\mds{T}$ on the space of physical maps, as induced by a suitable Riemannian metric.}\label{fig:Cover}
\end{figure}

To answer this, a universal measure that transcends specific underlying dynamics is required. Complexity is the natural notion for characterizing the difficulty of executing a task. Historically, the operational cost of completing specific tasks has been quantified using discrete metrics, such as algorithmic time steps for computation (computational complexity) \cite{Arora.2009} or primitive gate counts for unitary operations (gate complexity) \cite{Brandao.2021.PRXQ,Haferkamp.2022.NP}. However, to capture the continuous nature of physical processes, geometric complexity has recently emerged as a powerful framework \cite{Nielsen.2005.arxiv,Nielsen.2006.S,Dowling.2006.arxiv}. By measuring the geodesic length between the identity map and the target map within an abstract manifold, geometric complexity distills disparate physical constraints into a single, universal quantity. Despite the centrality of state-reset operations to thermodynamics, computation, and error correction, it has remained unknown whether a rigorous trade-off exists linking the complexity of a physical map to its execution fidelity.

In this work, we establish a complexity unattainability principle that bridges the thermodynamic and operational realms for dissipative maps (see Fig.~\ref{fig:Cover}). We consider the manifold of physical maps and introduce a Riemannian metric that is designed to satisfy three essential properties: physical penalization, protocol scaling, and resource divergence. By employing geometric complexity to quantify the cost of continuously implementing a reset map, we derive a fundamental trade-off relation between the geometric complexity of the map and its execution error [Eq.~\eqref{eq:geo.comp.3rd.law}]. This relation indicates that as the error of the reset map approaches zero, the required geometric complexity inevitably diverges, formally mirroring the thermodynamic unattainability of absolute zero. Crucially, this trade-off relation is universal: it holds strictly across the classical-quantum divide, governing both classical stochastic processes and quantum channels. Our results thus establish a fundamental principle independent of the underlying dynamics, imposing strict limits on basic operations in thermodynamics and computation.

\section*{Results}
We consider a stochastic map $\mds{T}$ on the space of $d$-dimensional probability distributions.
Specifically, the map $\mds{T}=[T_{mn}]\in\mbb{R}^{d\times d}$ is a matrix satisfying $T_{mn}\ge 0$ and $\sum_{m=1}^dT_{mn}=1$ for each $n=1,\dots,d$.
This stochastic map transforms any initial probability distribution $\vb{p}$ into a final distribution $\vb{p}'$ according to $\vb{p}'=\mds{T}\vb{p}$.
For example, when the underlying dynamics is a Markov jump process governed by a time-dependent transition matrix $\mds{W}_t$ over a total duration $\tau$, the map $\mds{T}$ is given by
\begin{equation}\label{eq:Markov.map}
	\mds{T}=\mca{T}e^{\int_0^\tau\dd{t}\mds{W}_t},
\end{equation}
where $\mca{T}$ denotes the time-ordering operator.
The stochastic maps described by Eq.~\eqref{eq:Markov.map} are important examples, but they are not exhaustive.
Notably, infinitely many stochastic maps cannot be realized by Markov jump processes.
For instance, $\mds{T}$ can represent classical non-Markovian dynamics, such as semi-Markov processes, in which the time-dependent distribution $\vb{p}_t$ evolves according to the generalized master equation \cite{Breuer.2008.PRL}
\begin{equation}
	\dot{\vb{p}}_t=\int_0^t\dd{s}\mds{W}_s\vb{p}_{t-s},
\end{equation}
where $\cdot$ denotes the time derivative.
More generally, a stochastic map can also be induced by a quantum channel $\Lambda$ by considering transitions between diagonal elements.
That is, $T_{mn}=\mel{m}{\Lambda(\dyad{n})}{m}$ for any $m$ and $n$, where $\{\ket{n}\}_{n\ge 1}$ is an orthonormal basis.
In this study, we do not focus on the underlying dynamics used to realize the map; rather, we focus on the stochastic map itself.
Although we primarily discuss discrete-state cases, the extension to continuous-state cases is straightforward and is presented in the Appendices.

Among all feasible maps, we are particularly interested in those that are designed to suppress the probabilities of specific states to zero.
For example, in information erasure, the target map should reset the probability of the logical state `1' to zero, regardless of the initial distribution.
In cooling, the populations of excited states should be reset to zero so that the system is driven to its ground state.
For the initialization of pure states, the reset map should reduce to zero the probability of finding the system in the subspace orthogonal to the target pure state.
In this work, we characterize the difficulty of implementing such reset maps from the perspective of complexity.

To characterize the complexity of a physical map, one possible quantifier is gate complexity.
In quantum computation, the gate complexity of a unitary circuit $U$ is defined as the minimum number of primitive gates (e.g., one- and two-qubit gates) required to implement $U$.
Although stochastic maps and unitary operators have different mathematical structures, it is still natural to ask whether an analogous notion of gate complexity can be defined for stochastic maps.
By analogy with local gates, we define primitive operations that act nontrivially only on two states:
\begin{equation}\label{eq:local.ops}
	\mds{T}_{mn}(\alpha,\beta)=\begin{pmatrix}
		1-\alpha & \beta\\
		\alpha & 1-\beta
	\end{pmatrix}\bigoplus\mds{1}_{\setminus\{m,n\}},
\end{equation}
where $\alpha,\beta\in[0,1]$ and $\mds{1}$ denotes the identity matrix.
This operation exchanges probability between states $m$ and $n$ according to $p_m\leftarrow (1-\alpha)p_m+\beta p_n$ and $p_n\leftarrow \alpha p_m+(1-\beta)p_n$.
Given these primitive operations, a natural question is whether there exists a finite sequence $\{\mds{T}_{m_kn_k}(\alpha_k,\beta_k)\}_{k\ge 1}$ such that the stochastic map $\mds{T}$ can be decomposed as
\begin{equation}\label{eq:decomp.map}
	\mds{T}\stackrel{?}{=}\prod_{k\ge 1}\mds{T}_{m_kn_k}(\alpha_k,\beta_k).
\end{equation}
However, although this decomposition is feasible for certain maps, there exist stochastic maps that cannot be decomposed in this form.
If each primitive operation is implemented by a finite-time Markov jump process, one can verify that the constraint $\alpha+\beta\le 1$ must hold.
Under this constraint, the impossibility of Eq.~\eqref{eq:decomp.map} follows from the fact that each local operation has a nonnegative determinant, $\det[\mds{T}_{mn}(\alpha,\beta)]=1-\alpha-\beta\ge 0$.
Hence, any stochastic map with a negative determinant cannot be decomposed as in Eq.~\eqref{eq:decomp.map}, because determinants multiply under composition: $\det(\mds{T})=\prod_{k\ge 1}\det[\mds{T}_{m_kn_k}(\alpha_k,\beta_k)]$. If a stochastic map has a positive determinant, swapping any two rows yields one with a negative determinant.
Even if the condition $\alpha+\beta\le 1$ is relaxed, one can still construct stochastic maps that are not decomposable in this way (see Appendix \ref{app:decomp.nogo} for the proof).
Therefore, it is not straightforward to define gate complexity for stochastic maps solely in terms of such primitive operations.

An alternative approach to quantifying complexity is grounded in geometry. Initially proposed for unitary operations in quantum computation \cite{Nielsen.2005.arxiv,Nielsen.2006.S,Dowling.2006.arxiv}, this framework has since been applied across diverse areas, including quantum many-body physics \cite{Balasubramanian.2020.JHEP,Balasubramanian.2021.JHEP}, quantum field theory \cite{Jefferson.2017.JHEP,Chapman.2018.PRL}, and holographic quantum gravity \cite{Brown.2016.PRL,Brown.2018.PRD} (see Ref.~\cite{Baiguera.2026.PR} for a review). Intuitively, geometric complexity captures the operational difficulty of a unitary map by evaluating the cost of a smooth path connecting the identity to the target map. It employs a tailored metric on the unitary group that stretches directions associated with resource-intensive operations, thereby placing complex unitaries geometrically farther away. Consequently, continuous path length replaces discrete gate counts, and the complexity of a target unitary is rigorously defined by the length of its geodesic---the shortest, optimal continuous path. In what follows, we move beyond unitary operations, introducing geometric complexity for general classical stochastic maps and quantum channels.

\subsection*{Geometric complexity of stochastic maps}
Consider the manifold of stochastic maps equipped with a Riemannian metric $g$.
The geometric complexity of a map $\mds{T}$ is defined as the geodesic length induced by $g$ between the identity map $\mds{1}$ and the target map $\mds{T}$:
\begin{equation}\label{eq:geo.com.def}
	\mca{C}(\mds{T})\coloneqq\min_{\{\mds{T}_t\}}\int_0^1\dd{t}\sqrt{g_{\mds{T}_t}(\dot{\mds{T}}_t,\dot{\mds{T}}_t)},
\end{equation}
where the minimum is taken over all continuous paths $\{\mds{T}_t\}_{0\le t\le 1}$ connecting the identity map to the given stochastic map (i.e., $\mds{T}_0=\mds{1}$ and $\mds{T}_1=\mds{T}$).
Intuitively, $\mca{C}$ characterizes the minimum path cost over all continuous implementations of the stochastic map $\mds{T}$ starting from the identity map.
In general, there are many possible choices of metric $g$ for defining complexity.
However, an appropriate metric and the corresponding complexity should satisfy three properties:
\begin{itemize}
	\item[(i)] {\it Physical penalization}: it should adequately penalize reset maps that are physically difficult to realize; 
	\item[(ii)] {\it Protocol scaling}: it should provide a linear lower bound on the number of protocols, thereby capturing the cumulative cost of map implementation; 
	\item[(iii)] {\it Divergent resources}: most importantly, it should capture the divergent physical resources required to achieve perfect reset maps.
\end{itemize}
For any $d\times d$ matrix $\mds{X}$, let $\vb{r}^{\mds{X}}\coloneqq \mds{X}\vb{1}/d$ denote its normalized row-sum vector, where $\vb{1}$ is the all-one vector.
Using this notation and incorporating properties (i)--(iii), we define the following Riemannian metric on the space of stochastic maps:
\begin{equation}\label{eq:Hess.metric}
	g_{\mds{T}}(\mds{X},\mds{Y})\coloneqq \sum_{n=1}^d\frac{r_n^{\mds{X}}r_n^{\mds{Y}}}{(r_n^{\mds{T}})^2},
\end{equation}
where $\mds{X}$ and $\mds{Y}$ are arbitrary elements of the tangent space.
Evidently, $g$ satisfies bilinearity, symmetry, and positive semidefiniteness.
Moreover, $g$ is a log-barrier Hessian metric, arising from the Hessian of the potential $\Phi(\mds{T})=-\sum_{n=1}^d\ln r_n^{\mds{T}}$.
Accordingly, this metric assigns a large weight when the map passes near the target region, where $r_n^{\mds{T}}\ll 1$ for $n\in\mfr{u}$ [cf.~property (i)], and $\mfr{u}\subset\{1,\dots,d\}$ denotes the set of undesired states whose probabilities are ideally zero after the transformation.
Although the geodesic distance induced by this metric is analytically intractable under the stochasticity constraint $\vb{1}^\top\vb{r}^{\mds{T}}=1$, we can still derive simple lower and upper bounds on the complexity (Appendix \ref{app:cla.proofs}):
\begin{equation}\label{eq:geo.comp.lub}
	\ell\le \mca{C}(\mds{T})\le (\sqrt{d}+1)\ell,
\end{equation}
where $\ell\coloneqq\sqrt{\sum_{n=1}^d(\ln d + \ln r_n^{\mds{T}})^2}$ is the geodesic distance induced by the metric $g$ in Eq.~\eqref{eq:Hess.metric} \emph{without} imposing the constraint.
The lower bound implies that the complexity of perfect reset maps diverges, since $r_n^{\mds{T}}=0$ for $n\in\mfr{u}$ and hence $\ell=\infty$.

To further justify this geometric quantity as a measure of complexity, we show its relation to the number of protocols required to implement a stochastic map [cf.~property (ii)].
To this end, consider a stochastic map $\mds{T}$ constructed from $N$ sequential protocols, where the $k$th protocol is realized by Markovian dynamics with a time-independent transition matrix $\mds{W}_{k}$ over a unit time interval.
The resulting stochastic map is
$\mds{T}=\prod_{k=1}^{N}e^{\mds{W}_{k}}$.
Because each protocol is subject to physical implementation constraints, it is natural to require that the transition rates in each matrix $\mds{W}_{k}$ are not arbitrarily large.
Accordingly, we assume that the maximum escape rate is always bounded above by $\gamma$.
Under this assumption, we prove that the geometric complexity provides a lower bound on the number of protocols (Appendix \ref{app:cla.proofs}):
\begin{equation}\label{eq:num.prot.lb}
	N\ge\frac{\mca{C}(\mds{T})}{(d+\sqrt{d})\ln(de^\gamma)}.
\end{equation}
This inequality provides an operational interpretation of geometric complexity, thereby clarifying the physical meaning of $\mca{C}$ as a complexity measure.
In Appendix \ref{app:cla.proofs}, we show that, within a parameterized class of Riemannian metrics that includes the classical Fisher information metric, the metric defined in Eq.~\eqref{eq:Hess.metric} is the unique one satisfying properties (i)--(iii).

In addition, we derive an entropic lower bound on geometric complexity, thereby clarifying its connection to thermodynamics.
Specifically, we prove that the complexity is always lower-bounded by the difference between the Shannon entropy of the uniform distribution before and after resetting by the map (Appendix \ref{app:cla.proofs}):
\begin{equation}\label{eq:comp.ent.rel}
	\mca{C}(\mds{T})\ge \qty\ln S(\vb{1}/d) - \ln S(\mds{T}\vb{1}/d),
\end{equation}
where $S(\vb{p})\coloneqq -\sum_{n=1}^dp_n\ln p_n$ is the Shannon entropy of the probability distribution $\vb{p}$.
Relations \eqref{eq:num.prot.lb} and \eqref{eq:comp.ent.rel} thus highlight the operational and thermodynamic aspects of geometric complexity.
Notably, relation \eqref{eq:comp.ent.rel} already sheds light on the unattainability of absolute zero entropy: achieving zero entropy, $S(\mds{T}\vb{1}/d)=0$, necessitates infinite complexity $\mca{C}$. Nonetheless, in what follows, we establish a universal principle for reset maps, even in scenarios where the entropy is not reset to zero.

\subsection*{Unattainability principle via geometric complexity}
With this continuous geometric framework in hand, we now investigate the fundamental difficulty of implementing state-reset maps.
The error $\epsilon$ associated with a reset map $\mds{T}$ is naturally defined as
\begin{equation}
	\epsilon(\mds{T})\coloneqq \sum_{m\in\mfr{u}}\sum_{n=1}^dT_{mn}=\sum_{m\in\mfr{u}}(\mds{T}\vb{1})_m.
\end{equation}
Since $T_{mn}$ is the conditional transition probability from state $n$ to state $m$, $\epsilon$ quantifies the extent to which the system remains in the undesired subspace $\mfr{u}$.
As our central result, we establish a rigorous trade-off relation between the geometric complexity of any stochastic map and its execution error:
\begin{equation}\label{eq:geo.comp.3rd.law}
	e^{\mca{C}(\mds{T})}\times\epsilon(\mds{T})\ge 1.
\end{equation}
This inequality dictates that achieving zero error, $\epsilon(\mds{T})=0$, requires the geometric complexity $\mca{C}(\mds{T})$ to diverge, thereby manifesting the third law of thermodynamics in the implementation of stochastic maps.
Therefore, this result can also be regarded as a \emph{geometric} third law of thermodynamics.
Together with relation \eqref{eq:num.prot.lb}, it also implies that infinitely many protocols are required to implement a perfect stochastic map via Markovian dynamics.

We illustrate the unattainability principle using a simple bit-reset map. Consider a reset map $\mds{T}$ implemented by a time-independent transition rate matrix $\mds{W}$ over a duration $\tau$, given by
\begin{equation}
	\mds{T}=e^{\mds{W}\tau},~\text{where}~\mds{W}=\begin{pmatrix}
		0 & w\\
		0 & -w
	\end{pmatrix}.
\end{equation}
Since the transition rate $w$ from the excited state to the ground state is positive and the reverse jump is forbidden, this map can bring the system close to the ground state regardless of the initial distribution.
From the error expression $\epsilon(\mds{T})=e^{-w\tau}$, it is evident that the key physical quantity governing reset accuracy is the product of the operation time and the kinetic transition rate.
For this stochastic map, the complexity can be calculated explicitly as
\begin{equation}\label{eq:two.lev.comp}
	\mca{C}(\mds{T})=2\atanh[\sqrt{2}f(w\tau)]-\sqrt{2}\atanh[f(w\tau)],
\end{equation}
where $f(u)=(1-e^{-u})/\sqrt{1+(1-e^{-u})^2}$.
From the asymptotic relation $\mca{C}(\mds{T})=w\tau + 2\ln 2 - \sqrt{2}\ln(1+\sqrt{2}) + O(e^{-w\tau})$, we conclude that $\mca{C}(\mds{T})$ captures both the operation time and the strength of the transition rates as relevant resources for reset maps.
Using $\epsilon(\mds{T})=e^{-w\tau}$ together with Eq.~\eqref{eq:two.lev.comp}, we can confirm, both analytically and numerically, that the complexity satisfies the trade-off relation \eqref{eq:geo.comp.3rd.law}.

In what follows, we provide a brief proof of the trade-off relation \eqref{eq:geo.comp.3rd.law}.
First, note that the error of the stochastic map $\mds{T}$ can be written as $\epsilon(\mds{T})=d\sum_{n\in\mfr{u}}r_n^{\mds{T}}$.
Since $e^{x}\ge e^{-|x|}$ for any real $x$ and $e^{-x}$ is convex, we obtain the following lower bound:
\begin{align}
	\epsilon(\mds{T})=\sum_{n\in\mfr{u}}e^{\ln(dr_n^{\mds{T}})}&\ge\sum_{n\in\mfr{u}}e^{-|\ln(dr_n^{\mds{T}})|}\notag\\
	&\ge|\mfr{u}|\,e^{-|\mfr{u}|^{-1}\sum_{n\in\mfr{u}}|\ln(dr_n^{\mds{T}})|},\label{eq:error.lb1}
\end{align}
where the last inequality follows from Jensen's inequality.
Next, by the Cauchy--Schwarz inequality, $\ell$ satisfies
\begin{equation}\label{eq:ell.lb1}
	\ell\ge\frac{1}{\sqrt{|\mfr{u}|}}\sum_{n\in\mfr{u}}|\ln(dr_n^{\mds{T}})|.
\end{equation}
Combining Eqs.~\eqref{eq:error.lb1} and \eqref{eq:ell.lb1}, we find
\begin{equation}
	\epsilon(\mds{T})\ge|\mfr{u}|e^{-\ell/\sqrt{|\mfr{u}|}}\ge e^{-\ell}.
\end{equation}
Finally, since $\ell\le\mca{C}(\mds{T})$ [Eq.~\eqref{eq:geo.comp.lub}], we immediately obtain the desired relation \eqref{eq:geo.comp.3rd.law}.

\subsection*{Unified quantum generalization}
Next, we derive a quantum generalization of geometric complexity and thereby establish the universality of the trade-off relation \eqref{eq:geo.comp.3rd.law} between complexity and error.
As a quantum analog of classical stochastic maps, we consider a quantum channel $\Lambda$, i.e., a completely positive trace-preserving (CPTP) map.
A relevant example is $\Lambda(\circ)=\tr_E[\mds{U}(\circ\otimes\pi_E)\mds{U}^\dagger]$, which describes the general reduced dynamics of the system after the environment $E$ is traced out. Here, $\pi_E$ denotes the thermal Gibbs state of the environment, and the unitary operator $\mds{U}$ acts jointly on the system and environment.

We consider a quantum channel that resets arbitrary initial states from an undesired subspace.
Because it is sufficient to evaluate the error for the maximally mixed state, we define the error as the probability that the system remains in the undesired subspace (scaled by a factor of $d$):
\begin{equation}
	\epsilon(\Lambda)\coloneqq \tr[\Pi\Lambda\qty(\mds{1})],
\end{equation}
where $\Pi$ is the projector onto the undesired subspace.
For example, for cooling to the ground state, $\Pi=\sum_{n:\,E_n>0}\dyad{E_n}$, where the sum runs over all excited energy eigenstates; for qubit erasure, $\Pi$ is the projector onto the logical-`1' subspace.
Since the error for an arbitrary initial state is always upper-bounded by $\epsilon(\Lambda)$, this quantity provides a reliable error measure for quantum channels.

To formalize the geometry of quantum channels, it is convenient to use the Choi-matrix representation \cite{Choi.1975.LAA,Jamiokowski.1972.RMP}, which provides a channel--state duality. For any continuous path of quantum channels $\{\Lambda_t\}$, we consider the corresponding path of bipartite density operators $\{\mds{M}_t\}$, where $\mds{M}_t=(\mds{1}\otimes\Lambda_t)(\dyad{\Omega})$ with $\ket{\Omega}=\sum_{n=1}^d\ket{n}\otimes\ket{n}$ for an orthonormal basis $\{\ket{n}\}$.
Then we define the geometric complexity of the channel $\Lambda$, associated with a metric $g$, as
\begin{equation}
	\mca{C}(\Lambda)\coloneqq\min_{\{\Lambda_t\}}\int_0^1\dd{t}\sqrt{g_{\mds{M}_t}(\dot{\mds{M}}_t,\dot{\mds{M}}_t)},
\end{equation}
where the minimum is taken over all paths connecting the identity channel to $\Lambda$ [i.e., $\Lambda_0=\msf{Id}$ and $\Lambda_1=\Lambda$, with $\msf{Id}(\circ)=\circ$].
The Riemannian metric $g$ on the space of the Choi matrices, which possesses three desired properties (i)--(iii), is defined by
\begin{equation}
	g_{\mds{M}}(\mds{X},\mds{Y})\coloneqq\tr(\phi_{\mds{M}}^{-1}\phi_{\mds{X}}\phi_{\mds{M}}^{-1}\phi_{\mds{Y}}),
\end{equation}
where $\phi_{\mds{X}}\coloneqq d^{-1}\tr_1\mds{X}$ for any bipartite operator $\mds{X}$.
In the classical limit, this geometric complexity and its associated Riemannian metric reduce exactly to their classical counterparts in Eqs.~\eqref{eq:geo.com.def} and \eqref{eq:Hess.metric}.
Notably, all geometric relations derived for the classical case carry over to the quantum case (Appendix \ref{app:qua.details}).
First, simple lower and upper bounds on the complexity, analogous to Eq.~\eqref{eq:geo.comp.lub}, can be obtained: $\ell\le \mca{C}(\Lambda)\le (\sqrt{d}+1)\ell$, where $\ell=\|\ln d+\ln\Lambda(\mds{1}/d)\|_F$ and $\|\mds{X}\|_F\coloneqq\sqrt{\tr(\mds{X}^\dagger\mds{X})}$ denotes the Frobenius norm of an operator $\mds{X}$.
Second, the relation \eqref{eq:num.prot.lb} between the complexity and the number of protocols can also be established for both Lindbladian maps and general dissipative maps.
Finally, the entropic bound \eqref{eq:comp.ent.rel} on geometric complexity can be derived analogously, with the Shannon entropy replaced by the von Neumann entropy.

Based on this quantum geometric framework, we arrive at the same complexity-error trade-off relation for quantum channels:
\begin{equation}\label{eq:qua.comp.3rd.law}
	e^{\mca{C}(\Lambda)}\times\epsilon(\Lambda)\ge 1.
\end{equation}
This structural parallelism between the classical and quantum domains highlights the universality of the geometric unattainability principle.

\section*{Conclusion}
Understanding the fundamental physical limits of thermodynamic and computational processes requires a unified framework that subsumes diverse operational resources---such as execution time, energetic cost, and control bandwidth---into a single measure. In this work, we showed that complexity provides such a unifying measure, extending the third law of thermodynamics from a constraint on state-to-state cooling to a universal limit on physical map implementation. Specifically, by leveraging a geometric framework to quantify operational burden, we developed a comprehensive notion of complexity for both classical stochastic maps and quantum channels. This approach enabled us to establish a universal trade-off relation between the geometric complexity of a state-reset map and its execution error. By subsuming disparate physical ingredients into a single quantity, our central result yields a dynamics-independent unattainability principle: a perfect state-reset operation can be physically implemented only at diverging operational complexity. More broadly, by formalizing the relation between operational complexity and execution error, we showed that the third law of thermodynamics is not merely a state-dependent constraint on entropy extraction, but a fundamental and universal complexity bound of nature.

\section*{Acknowledgments}
{T.V.V.} was supported by JSPS KAKENHI Grants No.~JP23K13032, No.~JP26K00022, and No.~JP26H02015.
{K.S.} was supported by JSPS KAKENHI Grants No.~JP23K25796, No.~JP26H02015, and No.~JP26H00388.

\section*{Data availability}
No data were created or analyzed in this work.

\appendix

\section{Proof of the impossibility of decomposing stochastic maps}\label{app:decomp.nogo}
\begin{lemma}
Consider primitive operations of the following form:
\begin{equation}
	\mds{T}_{mn}(\alpha,\beta)=\begin{pmatrix}
		1-\alpha & \beta\\
		\alpha & 1-\beta
	\end{pmatrix} \otimes \mds{1}_{\setminus \{m,n\}},
\end{equation}
where $\alpha,\beta\in[0,1]$.
Then there exists a stochastic map that cannot be decomposed into a product of these primitive operations.
\end{lemma}
\begin{proof}
It suffices to provide a counterexample.
Specifically, we show that the following simple stochastic map cannot be decomposed into the given primitive operations:
\begin{equation}
	\mds{T}=\begin{pmatrix}
		0 & 1/2 & 1/2\\
		1/2 & 0 & 1/2\\
		1/2 & 1/2 & 0
	\end{pmatrix}.
\end{equation}
To this end, suppose that $\mds{T}$ can be decomposed into a product of primitive operations:
\begin{equation}
	\mds{T}=\prod_{k=1}^{N}\mds{T}_{m_kn_k}(\alpha_k,\beta_k),
\end{equation}
where each $\mds{T}_{m_kn_k}(\alpha_k,\beta_k)$ is a nontrivial primitive operation [i.e., $\mds{T}_{m_kn_k}(\alpha_k,\beta_k)\neq\mds{1}$, or equivalently, $(\alpha_k,\beta_k)\neq(0,0)$] that acts only on the $m_k$th and $n_k$th rows for any $k=1,\dots,N$.
Let $\imath\in[1,N]$ be the maximal index such that $(\alpha_{\imath},\beta_{\imath})\neq(1,1)$.
In other words, $\mds{T}_{m_kn_k}(\alpha_k,\beta_k)$ is a swap operation for any $\imath+1\le k\le N$.
Note that there always exist permutation matrices $\mds{P}_1$ and $\mds{P}_2$ (which are trivially stochastic maps) such that $\mds{T}_{m_{\imath}n_{\imath}}(\alpha_{\imath},\beta_{\imath})$ can be written as
\begin{align}
	\mds{T}_{m_{\imath}n_{\imath}}(\alpha_{\imath},\beta_{\imath})&=\mds{P}_1\begin{pmatrix}
		1-\alpha_{\imath} & \beta_{\imath} & 0\\
		\alpha_{\imath} & 1-\beta_{\imath} & 0\\
		0 & 0 & 1
	\end{pmatrix}\mds{P}_2\notag\\
	&=\mds{P}_1\mds{T}_{12}(\alpha_{\imath},\beta_{\imath})\mds{P}_2.
\end{align}
Define $\mds{U}\coloneqq[\prod_{k=\imath+1}^N\mds{T}_{m_kn_k}(\alpha_k,\beta_k)]\mds{P}_1$ and $\mds{V}\coloneqq\mds{P}_2\prod_{k=1}^{\imath-1}\mds{T}_{m_kn_k}(\alpha_k,\beta_k)$.
Both are valid stochastic maps, and it follows that $\mds{T}=\mds{U}\mds{T}_{12}(\alpha_{\imath},\beta_{\imath})\mds{V}$.
Note that $\mds{P}^\top\mds{P}=\mds{P}\mds{P}^\top=\mds{1}$ for any permutation matrix $\mds{P}$, and that the product of permutation matrices is also a permutation matrix.
Combining this with the fact that $\mds{T}_{m_kn_k}(\alpha_k,\beta_k)$ is a permutation matrix for any $k\in[\imath+1,N]$ implies that $\mds{U}$ is a permutation matrix.
Since $\mds{U}^\top\mds{T}$ is obtained from $\mds{T}$ by swapping rows, $\mds{U}^\top\mds{T}$ can be transformed back to $\mds{T}$ by swapping columns (due to the symmetry of $\mds{T}$).
Equivalently, there exists a permutation matrix $\mds{U}'$ (which is also a stochastic map) such that $\mds{U}^\top\mds{T}\mds{U}'=\mds{T}$.
Substituting $\mds{T}=\mds{U}\mds{T}_{12}(\alpha_{\imath},\beta_{\imath})\mds{V}$ into the left-hand side of this equality, we obtain $\mds{T}=\mds{T}_{12}(\alpha_{\imath},\beta_{\imath})\mds{V}\mds{U}'=\mds{T}_{12}(\alpha_{\imath},\beta_{\imath})\mds{O}$, where $\mds{O}\coloneqq \mds{V}\mds{U}'$ is a stochastic map.
The relation $\mds{T}=\mds{T}_{12}(\alpha_{\imath},\beta_{\imath})\mds{O}$ can be explicitly expressed as
\begin{widetext}
\begin{align}
	\begin{pmatrix}
		0 & 1/2 & 1/2\\
		1/2 & 0 & 1/2\\
		1/2 & 1/2 & 0
	\end{pmatrix} &= \begin{pmatrix}
		1-\alpha_{\imath} & \beta_{\imath} & 0\\
		\alpha_{\imath} & 1-\beta_{\imath} & 0\\
		0 & 0 & 1
	\end{pmatrix} \begin{pmatrix}
		O_{11} & O_{12} & O_{13}\\
		O_{21} & O_{22} & O_{23}\\
		O_{31} & O_{32} & O_{33}
	\end{pmatrix} \notag\\
	&= \begin{pmatrix}
		(1-\alpha_{\imath}) O_{11} + \beta_{\imath} O_{21} & (1-\alpha_{\imath}) O_{12} + \beta_{\imath} O_{22} & (1-\alpha_{\imath}) O_{13} + \beta_{\imath} O_{23}\\
		\alpha_{\imath} O_{11} + (1-\beta_{\imath}) O_{21} & \alpha_{\imath} O_{12} + (1-\beta_{\imath}) O_{22} & \alpha_{\imath} O_{13} + (1-\beta_{\imath})O_{23}\\
		O_{31} & O_{32} & O_{33}
	\end{pmatrix}.\label{eq:decomp.tmp1}
\end{align}
\end{widetext}
Note that $(\alpha_{\imath},\beta_{\imath})\neq (1,1)$ and $(\alpha_{\imath},\beta_{\imath})\neq (0,0)$.
If $\alpha_{\imath}\in[0,1)$ and $\beta_{\imath}\in(0,1]$, then the condition $0=(1-\alpha_{\imath}) O_{11} + \beta_{\imath} O_{21}$ immediately implies $O_{11}=O_{21}=0$.
As a result, $\alpha_{\imath} O_{11} + (1-\beta_{\imath}) O_{21}=0\neq 1/2$, which violates Eq.~\eqref{eq:decomp.tmp1}.
Therefore, either $\alpha_{\imath}=1$ or $\beta_{\imath}=0$.
If $\alpha_{\imath}=1$, the conditions $0=\alpha_{\imath} O_{12} + (1-\beta_{\imath}) O_{22}$ and $(\alpha_{\imath},\beta_{\imath})\neq(1,1)$ imply $O_{12}=O_{22}=0$; thus, $(1-\alpha_{\imath}) O_{12} + \beta_{\imath} O_{22}=0\neq 1/2$, which violates Eq.~\eqref{eq:decomp.tmp1}.
Otherwise, if $\beta_{\imath}=0$, the conditions $0=\alpha_{\imath} O_{12} + (1-\beta_{\imath}) O_{22}$ and $(\alpha_{\imath},\beta_{\imath})\neq(0,0)$ lead to $O_{12}=O_{22}=0$; thus, $(1-\alpha_{\imath}) O_{12} + \beta_{\imath} O_{22}=0\neq 1/2$, which also violates Eq.~\eqref{eq:decomp.tmp1}.
In other words, a contradiction always arises.
Therefore, $\mds{T}$ cannot be decomposed using the given primitive operations, which completes the proof.
\end{proof}

\section{Proof of results in the classical case}\label{app:cla.proofs}
\subsection{Proof of the lower and upper bounds \eqref{eq:geo.comp.lub}}
We first prove the lower bound.
The geodesic distance under the constraint $\vb{1}^\top\vb{r}=1$ is always larger than or equal to that without the constraint.
Therefore, the geometric complexity can be evaluated as
\begin{equation}
	\mca{C}(\mds{T})\ge\min_{\{\vb{r}_t\}_{0\le t\le 1}}\int_0^1\dd{t}\sqrt{\sum_{n=1}^d\qty(\frac{\dot r_n(t)}{r_n(t)})^2},
\end{equation}
where the minimum is over all paths $\{\vb{r}_t\}_{0\le t\le 1}$ connecting $\vb{r}^{\mds{1}}$ and $\vb{r}^{\mds{T}}$.
Noting that the geodesic path has constant speed and applying the Cauchy--Schwarz inequality, the lower bound is derived as follows:
\begin{align}
	\mca{C}(\mds{T})&\ge\qty[\min_{\{\vb{r}_t\}_{0\le t\le 1}}\int_0^1\dd{t}\sum_{n=1}^d\qty(\frac{\dot r_n(t)}{r_n(t)})^2]^{1/2}\notag\\
	&=\qty[\min_{\{\vb{r}_t\}_{0\le t\le 1}}\sum_{n=1}^d\int_0^1\dd{t}\qty(\dv{\ln r_n(t)}{t})^2]^{1/2}\notag\\
	&\ge \qty[\min_{\{\vb{r}_t\}_{0\le t\le 1}}\sum_{n=1}^d\qty(\int_0^1\dd{t}\dv{\ln r_n(t)}{t})^2]^{1/2}\notag\\
	&=\qty[\sum_{n=1}^d\qty(\ln r_n^{\mds{T}}-\ln r_n^{\mds{1}})^2]^{1/2}\notag\\
	&=\ell.
\end{align}

Next, we prove the upper bound $\mca{C}(\mds{T})\le(\sqrt{d}+1)\ell$.
We consider a specific path $\{\mds{T}_t\}_{0\le t\le 1}$ connecting $\mds{1}$ and $\mds{T}$ that satisfies
\begin{equation}
	\ln r_n^{\mds{T}_t}=\ln r_n^{\mds{1}} + (\ln r_n^{\mds{T}} - \ln r_n^{\mds{1}})t + c_t,
\end{equation}
where $c_t$ is a constant chosen to satisfy the constraint $\vb{1}^\top \vb{r}^{\mds{T}_t}=1$.
Specifically, $c_t$ is defined by
\begin{equation}
	c_t\coloneqq -\ln\qty[\sum_{n=1}^d e^{\ln r_n^{\mds{1}}+(\ln r_n^{\mds{T}} - \ln r_n^{\mds{1}})t}].
\end{equation}
As proved in Theorem \ref{theo:smap.exist} (Appendix \ref{app:cla.details}), such a path of stochastic maps always exists.
The time derivative of $c_t$ is
\begin{equation}
	\dot c_t=-\frac{\sum_{n=1}^d (\ln r_n^{\mds{T}} - \ln r_n^{\mds{1}}) e^{\ln r_n^{\mds{1}}+(\ln r_n^{\mds{T}} - \ln r_n^{\mds{1}})t}}{\sum_{n=1}^d e^{\ln r_n^{\mds{1}}+(\ln r_n^{\mds{T}} - \ln r_n^{\mds{1}})t}}.
\end{equation}
For this path, the complexity is upper bounded as
\begin{align}
	\mca{C}(\mds{T})&\le \int_0^1\dd{t}\sqrt{\sum_{n=1}^d\qty(\dv{\ln r_n^{\mds{T}_t}}{t})^2}\notag\\
	&=\int_0^1\dd{t}\sqrt{\sum_{n=1}^d(\ln r_n^{\mds{T}} - \ln r_n^{\mds{1}}+\dot{c}_t)^2}\notag\\
	&\le \int_0^1\dd{t}\qty[ \sqrt{\sum_{n=1}^d(\ln r_n^{\mds{T}} - \ln r_n^{\mds{1}})^2} + \sqrt{d(\dot c_t)^2}].\label{eq:C.ub.tmp1}
\end{align}
Here, we use the Minkowski inequality to obtain the last line.
Applying the inequality $|\sum_nx_ny_n/\sum_ny_n| \le \max_n|x_n|$ for any real numbers $\{x_n\}$ and $\{y_n\}$ satisfying $y_n>0~\forall n$, we obtain
\begin{equation}\label{eq:dotc.ub}
	|\dot c_t|\le\max_{1\le n\le d}|\ln r_n^{\mds{T}} - \ln r_n^{\mds{1}}|\le \ell.
\end{equation}
Combining Eqs.~\eqref{eq:C.ub.tmp1} and \eqref{eq:dotc.ub} yields the desired upper bound:
\begin{equation}
	\mca{C}(\mds{T})\le (\sqrt{d}+1)\ell.
\end{equation}

\subsection{Proof of the protocol-scaling relation \eqref{eq:num.prot.lb}}
Consider the stochastic map $\{\mds{T}_{t}\}_{0\le t\le N}$ generated by the sequence of transition rate matrices $\{\mds{W}_k\}_{k=1}^N$.
That is, $\mds{T}_t=\mca{T}e^{\int_0^t\dd{s}\mds{W}_s}$, where $\mds{W}_t=\mds{W}_{k}$ for $t\in[k-1,k)$ and $1\le k\le N$.
Evidently, $\mds{T}_0=\mds{1}$ and $\mds{T}_N=\mds{T}$.
The time evolution of $\vb{r}_t\coloneqq\vb{r}^{\mds{T}_t}$ is given by $\dot{\vb{r}}_t=\mds{W}_t\vb{r}_t$, from which we obtain
\begin{align}
	-\frac{\dot r_n(t)}{r_n(t)}&=-W_{nn}(t)-\frac{1}{r_n(t)}\sum_{m(\neq n)}W_{nm}(t)r_m(t)\notag\\
	&\le -W_{nn}(t).
\end{align}
Integrating this inequality over time, we obtain
\begin{equation}\label{eq:lndiff.ub1}
	-\ln r_n^{\mds{T}} + \ln r_n^{\mds{1}} \le -\int_0^N\dd{t}W_{nn}(t)\le N\gamma.
\end{equation}
On the other hand, since $r_n^{\mds{T}}\in[0,1]$, we have
\begin{equation}\label{eq:lndiff.ub2}
	\ln r_n^{\mds{T}} - \ln r_n^{\mds{1}}\le -\ln r_n^{\mds{1}}=\ln d.
\end{equation}
Combining Eqs.~\eqref{eq:lndiff.ub1} and \eqref{eq:lndiff.ub2}, we obtain
\begin{equation}
	|\ln r_n^{\mds{T}} - \ln r_n^{\mds{1}}| \le \max(N\gamma,\ln d) \le N\ln(de^\gamma).
\end{equation}
Consequently,
\begin{equation}
	\ell=\sqrt{\sum_{n=1}^d|\ln r_n^{\mds{T}} - \ln r_n^{\mds{1}}|^2}\le N\sqrt{d}\ln(de^\gamma).
\end{equation}
Combining this with the inequality $\mca{C}(\mds{T})\le (\sqrt{d}+1)\ell$ shown in Eq.~\eqref{eq:geo.comp.lub}, we obtain
\begin{equation}
	N\ge \frac{\mca{C}(\mds{T})}{(d+\sqrt{d})\ln(de^\gamma)},
\end{equation}
which is exactly Eq.~\eqref{eq:num.prot.lb}.

\subsection{Proof of the entropic bound \eqref{eq:comp.ent.rel}}
Let $\{\mds{T}_t\}_{0\le t\le 1}$ be the geodesic path, and define $\vb{r}_t\coloneqq\vb{r}^{\mds{T}_t}$. Then, the geometric complexity is evaluated as follows:
\begin{align}
	\mca{C}(\mds{T})&=\int_0^1\dd{t}\sqrt{\sum_{n=1}^d\qty[\frac{\dot r_n(t)}{r_n(t)}]^2}\notag\\
	&\ge\int_0^1\dd{t}\sqrt{\frac{[\sum_{n=1}^d\dot r_n(t)\ln r_n(t)]^2}{\sum_{n=1}^d [r_n(t)\ln r_n(t)]^2}}\notag\\
	&\ge\int_0^1\dd{t}\qty|\dv{t}\ln S(\vb{r}_t)|\notag\\
	&\ge |\ln S(\vb{r}_0) - \ln S(\vb{r}_1)|\notag\\
	&=\ln S(\vb{1}/d) - \ln S(\mds{T}\vb{1}/d).
\end{align}
Here, we use the Cauchy--Schwarz inequality $\sum_na_n^2/b_n\ge(\sum_na_n)^2/\sum_nb_n$ for any nonnegative $\{b_n\}$ and real $\{a_n\}$ to obtain the second line, the inequality $\sum_{n=1}^d(r_n\ln r_n)^2\le (\sum_{n=1}^dr_n\ln r_n)^2 = S(\vb{r})^2$ to obtain the third line, and the triangle inequality $\int_0^1\dd{t}|f(t)|\ge|\int_0^1\dd{t}f(t)|$ for any function $f(t)$ to obtain the fourth line.

\subsection{Geometric complexity of doubly stochastic maps and unital quantum channels}
A direct consequence of our geometric framework is that the complexity of any doubly stochastic classical map $\mds{T}$, as well as any unital quantum channel $\Lambda$, vanishes ($\mca{C}=0$). This is not merely a mathematical convention; it is physically necessary for quantifying the cost of thermodynamic control. To see this, consider the fundamental operational goal of a reset map: initializing an unknown mixed system into a pure state. Such a process necessarily requires entropy extraction from the system. Although the same logic applies to the quantum case, we focus on the classical setting for clarity. A doubly stochastic map $\mds{T}$ satisfies both $\mds{T}^\top\vb{1}=\vb{1}$ (probability preservation) and $\mds{T}\vb{1}=\vb{1}$ (uniform-state preservation). For such maps, entropy is always nondecreasing for any initial probability distribution $\vb{p}$:
\begin{equation}
	S(\vb{p})\le S(\mds{T}\vb{p}).
\end{equation}
Because doubly stochastic maps cannot reduce entropy for any state, they are ineffective for state initialization. In the resource theory of purity (nonuniformity) \cite{Gour.2015.PR}, they are classified as free operations, as they cannot increase purity from a free state (i.e., the uniform distribution). Therefore, any rigorous metric quantifying the operational and thermokinetic burden of resetting should assign zero cost to entropy-increasing processes. By assigning zero geometric complexity to doubly stochastic maps, our framework correctly captures their operational equivalence to the identity map in the context of state reset.

\subsection{Uniqueness of the metric within a parameterized class of Riemannian metrics}
Here, we show that the defined metric is the unique metric satisfying all three properties within a parameterized class of Riemannian metrics that includes the classical Fisher information metric.
Consider a general Riemannian metric parameterized by $\alpha>0$:
\begin{equation}
	g_{\mds{T},\alpha}(\mds{X},\mds{Y})\coloneqq\sum_{n=1}^d\frac{r_n^{\mds{X}}r_n^{\mds{Y}}}{(r_n^{\mds{T}})^\alpha}.
\end{equation}
The metric $g_{\mds{T},2}$ corresponds to the metric employed in the definition of complexity, whereas $g_{\mds{T},1}$ is precisely the classical Fisher information metric.
In what follows, we show that only $\alpha=2$ yields a metric satisfying all three properties.
Specifically, although the metric $g_{\mds{T},\alpha}$ satisfies property (i) for all $\alpha\in(0,+\infty)$, property (ii) cannot be satisfied for $\alpha\in(2,+\infty)$, and property (iii) cannot be satisfied for $\alpha\in(0,2)$.
To this end, we first show that, for any $\alpha\in(0,2)\cup(2,+\infty)$,
\begin{equation}\label{eq:lub.C.alpha}
	\ell_\alpha\le\mca{C}_\alpha(\mds{T})\le(\sqrt{d}+1)\ell_\alpha,
\end{equation}
where $\mca{C}_\alpha$ is the complexity associated with the metric $g_{\mds{T},\alpha}$, and $\ell_\alpha$ denotes the geodesic distance without the stochasticity constraint, given by
\begin{equation}
	\ell_\alpha\coloneqq \frac{1}{|\alpha'|}\sqrt{\sum_{n=1}^d|(r_n^{\mds{T}})^{\alpha'} - (r_n^{\mds{1}})^{\alpha'}|^2}.
\end{equation}
Here, $\alpha'\coloneqq 1-\alpha/2\neq 0$.
The lower bound follows immediately from the following expression for the complexity:
\begin{equation}
	\mca{C}_\alpha(\mds{T})=\min_{\{\mds{T}_t\}}\int_0^1\dd{t}\sqrt{\sum_{n=1}^d\qty[\frac{1}{\alpha'}\dv{t}(r_n^{\mds{T}_t})^{\alpha'}]^2}.
\end{equation}
To derive the upper bound, we consider a specific path $\{\mds{T}_t\}_{0\le t\le 1}$ connecting $\mds{1}$ and $\mds{T}$ and satisfying
\begin{equation}\label{eq:uni.c.def}
	(r_n^{\mds{T}_t})^{\alpha'}=(r_n^{\mds{1}})^{\alpha'} + [(r_n^{\mds{T}})^{\alpha'} - (r_n^{\mds{1}})^{\alpha'}]t + c_t,
\end{equation}
where $c_t$ is a constant chosen to satisfy the constraint $\vb{1}^\top \vb{r}^{\mds{T}_t}=1$.
Such a path always exists by Theorem \ref{theo:smap.exist}.
For this path, the complexity is bounded from above as follows:
\begin{align}
	\mca{C}_\alpha(\mds{T})&\le \frac{1}{|\alpha'|}\int_0^1\dd{t}\sqrt{\sum_{n=1}^d[(r_n^{\mds{T}})^{\alpha'} - (r_n^{\mds{1}})^{\alpha'}+\dot{c}_t]^2}\notag\\
	&\le \frac{1}{|\alpha'|}\int_0^1\dd{t}\qty[ \sqrt{\sum_{n=1}^d[(r_n^{\mds{T}})^{\alpha'} - (r_n^{\mds{1}})^{\alpha'}]^2} + \sqrt{d(\dot c_t)^2}].\label{eq:uni.C.ub.tmp1}
\end{align}
Here, we use the Minkowski inequality to obtain the second line.
From Eq.~\eqref{eq:uni.c.def}, we also have
\begin{align}
	|\dot{c}_t|&=\qty|\frac{\sum_{n=1}^d[(r_n^{\mds{T}})^{\alpha'} - (r_n^{\mds{1}})^{\alpha'}](r_n^{\mds{T}_t})^{1-\alpha'}}{\sum_{n=1}^d(r_n^{\mds{T}_t})^{1-\alpha'}}|\notag\\
	&\le \max_{1\le n\le d}|(r_n^{\mds{T}})^{\alpha'} - (r_n^{\mds{1}})^{\alpha'}|\notag\\
	&\le |\alpha'|\ell_\alpha.\label{eq:uni.dotc.ub}
\end{align}
Here, we use the inequality $|\sum_n x_n y_n/\sum_n y_n| \le \max_n|x_n|$ for any real numbers $\{x_n\}$ and $\{y_n\}$ satisfying $y_n>0$ for all $n$ to obtain the second line.
Combining Eqs.~\eqref{eq:uni.C.ub.tmp1} and \eqref{eq:uni.dotc.ub}, we obtain the desired upper bound:
\begin{equation}
	\mca{C}_\alpha(\mds{T})\le (\sqrt{d}+1)\ell_\alpha.
\end{equation}

We are now ready to show that property (ii) cannot be satisfied for $\alpha\in(2,+\infty)$.
Consider the map $\mds{T}=e^{\mds{W}N}$, which is realized by applying the same protocol $N$ times, each generated by the transition rate matrix $\mds{W}$ over a unit time interval.
The transition rate matrix $\mds{W}$ is given by
\begin{equation}
	\mds{W}=\begin{pmatrix}
		0 & 1\\
		0 & -1
	\end{pmatrix}.
\end{equation}
It is straightforward to verify that
\begin{equation}
	 \mds{T}=\begin{pmatrix}
	 	1 & 1-e^{-N}\\
	 	0 & e^{-N}
	 \end{pmatrix}
\end{equation}
and that $\vb{r}^{\mds{T}}=[1-e^{-N}/2,e^{-N}/2]^\top$.
Applying the lower bound on the complexity in Eq.~\eqref{eq:lub.C.alpha}, we obtain
\begin{equation}
	\mca{C}_\alpha(\mds{T})\ge\frac{1}{|\alpha'|}|(e^{-N}/2)^{\alpha'}-(1/2)^{\alpha'}|.
\end{equation}
For $\alpha>2$, we have $\alpha'=1-\alpha/2<0$.
Therefore,
\begin{align}
	\mca{C}_\alpha(\mds{T})\ge\frac{2^{|\alpha'|}}{|\alpha'|}(e^{N|\alpha'|}-1).
\end{align}
This lower bound on $\mca{C}_\alpha(\mds{T})$ is exponential in $N$; therefore, $\mca{C}_\alpha(\mds{T})$ cannot yield a linear lower bound on $N$ of the form given in Eq.~\eqref{eq:num.prot.lb}. 
Hence, property (ii) cannot be satisfied for $\alpha>2$.

Next, we show that property (iii) cannot be satisfied for $\alpha\in(0,2)$.
In this case, $\alpha'=1-\alpha/2>0$.
From the upper bound in Eq.~\eqref{eq:lub.C.alpha}, we obtain
\begin{equation}
	\mca{C}_\alpha(\mds{T})\le\frac{\sqrt{d}+1}{\alpha'}\sqrt{\sum_{n=1}^d|(r_n^{\mds{T}})^{\alpha'} - (1/d)^{\alpha'}|^2}.
\end{equation}
It follows that
\begin{equation}
	\mca{C}_\alpha(\mds{T})\le\frac{d+\sqrt{d}}{\alpha'}\max(1-d^{-\alpha'},d^{-\alpha'}).
\end{equation}
Therefore, the complexity $\mca{C}_\alpha(\mds{T})$ remains finite even when the error of the map $\mds{T}$ is zero, which violates property (iii).

\section{Detailed derivations for the classical case}\label{app:cla.details}

For convenience, we denote the spaces of $d$-dimensional probability distributions and stochastic maps by $\mca{P}_d$ and $\mca{M}_d$, respectively:
\begin{align}
	\mca{P}_d&=\{\vb{p}\in\mbb{R}_{\ge 0}^d\,|\,\vb{1}^\top\vb{p}=1\},\\
	\mca{M}_d&=\{\mds{T}\in\mbb{R}_{\ge 0}^{d\times d}\,|\,\mds{T}^\top\vb{1}=\vb{1}\}.
\end{align}
The subspace of strictly positive probability distributions is defined as
\begin{equation}
	\mca{P}_d^+=\{\vb{p}\in\mbb{R}_{>0}^d\,|\,\vb{1}^\top\vb{p}=1\}.
\end{equation}

\subsection{Proof of the existence of a path of stochastic maps}

We prove the following theorem.
\begin{theorem}\label{theo:smap.exist}
Let $\mds{T}$ be a stochastic map such that $\vb{r}^{\mds{T}}=\mds{T}\vb{1}/d\in\mca{P}_d^+$, and let $\{\vb{r}_t\}_{0\le t\le 1}$ be a smooth path in $\mca{P}_d^+$ connecting $\vb{r}^{\mds{1}}$ and $\vb{r}^{\mds{T}}$.
Then there exists a corresponding continuous path $\{\mds{T}_t\}_{0\le t\le 1}$ in $\mca{M}_d$ connecting $\mds{1}$ and $\mds{T}$ such that $\mds{T}_t\vb{1}/d=\vb{r}_t$ for all $t\in[0,1]$.
\end{theorem}

\begin{proof}
First, consider a smooth path of stochastic maps, $\{\mds{A}_t\}_{0\le t\le 1}$, defined by
\begin{equation}
	\mds{A}_t = a\sin(\pi t)\frac{\vb{1}\vb{1}^\top}{d} + [1-a\sin(\pi t)][(1-t)\mds{1} + t\mds{T}],
\end{equation}
where $0<a<1$ is an arbitrary constant.
Evidently, each $\mds{A}_t$ is a valid stochastic map, $\mds{A}_0=\mds{1}$, $\mds{A}_1=\mds{T}$, and all elements of $\mds{A}_t$ are strictly positive for $t\in(0,1)$.
However, this path of stochastic maps may not satisfy the required condition $\mds{A}_t\vb{1}/d=\vb{r}_t$ for all $t\in[0,1]$.
To construct the desired path, we consider the following path of double-scaled stochastic maps:
\begin{equation}\label{eq:scale.map}
	\mds{T}_t=\diag(e^{\vb{u}_t})\mds{A}_t\diag(e^{\vb{v}_t}),
\end{equation}
where $\vb{u}_t$ and $\vb{v}_t$ are $d$-dimensional scaling vectors.
The scaling in Eq.~\eqref{eq:scale.map} is also known as the generalized Sinkhorn transformation \cite{Sinkhorn.1964.AMS}.
According to Proposition \ref{prop:scale.sol}, there exists a path $\{(\vb{u}_t,\vb{v}_t)\}_{0<t<1}$ such that $\mds{T}_t\vb{1}/d=\vb{r}_t$, $\mds{T}_t^\top\vb{1}=\vb{1}$, and $\vb{1}^\top\vb{v}_t=0$.
The gauge condition $\vb{1}^\top\vb{v}_t=0$ resolves the translational invariance of the log-linear scaling (i.e., $\vb{u}_t\leftarrow\vb{u}_t+c$ and $\vb{v}_t\leftarrow\vb{v}_t-c$ yield the same map $\mds{T}_t$ for any real $c$).
By the implicit function theorem, the path $\{(\vb{u}_t,\vb{v}_t)\}_{0<t<1}$ is smooth; consequently, the corresponding path of stochastic maps $\{\mds{T}_t\}_{0<t<1}$ is continuous.

The remaining task is to verify the boundary conditions, namely,
\begin{equation}
	\mds{T}_0=\mds{1}~\text{and}~\mds{T}_1=\mds{T},
\end{equation}
where $\mds{T}_{0/1}\coloneqq\lim_{t\to 0/1}\mds{T}_t$.
Let $\vb{u}_{0/1}=\lim_{t\to 0/1}\vb{u}_t$ and $\vb{v}_{0/1}=\lim_{t\to 0/1}\vb{v}_t$.
Then $\mds{T}_0=\diag(e^{\vb{u}_0})\mds{A}_0\diag(e^{\vb{v}_0})$ and $\mds{T}_1=\diag(e^{\vb{u}_1})\mds{A}_1\diag(e^{\vb{v}_1})$.
Note that both maps $\mds{T}_0=\diag(e^{\vb{u}_0})\mds{A}_0\diag(e^{\vb{v}_0})$ and $\mds{1}=\diag(e^{\vb{0}})\mds{A}_0\diag(e^{\vb{0}})$ satisfy
\begin{align}
	\mds{T}_0\vb{1}&=\lim_{t\to 0}\mds{T}_t\vb{1}=d\lim_{t\to 0}\vb{r}_t=d\vb{r}_0=\mds{1}\vb{1},\\
	\mds{T}_0^\top\vb{1}&=\lim_{t\to 0}\mds{T}_t^\top\vb{1}=\vb{1}=\mds{1}^\top\vb{1}.
\end{align}
Applying Proposition \ref{prop:equiv.map}, we obtain $\mds{T}_0=\mds{1}$.
Similarly, both maps $\mds{T}_1=\diag(e^{\vb{u}_1})\mds{A}_1\diag(e^{\vb{v}_1})$ and $\mds{T}=\diag(e^{\vb{0}})\mds{A}_1\diag(e^{\vb{0}})$ satisfy
\begin{align}
	\mds{T}_1\vb{1}&=\lim_{t\to 1}\mds{T}_t\vb{1}=d\lim_{t\to 1}\vb{r}_t=d\vb{r}_1=\mds{T}\vb{1},\\
	\mds{T}_1^\top\vb{1}&=\lim_{t\to 1}\mds{T}_t^\top\vb{1}=\vb{1}=\mds{T}^\top\vb{1}.
\end{align}
Again, Proposition \ref{prop:equiv.map} gives $\mds{T}_1=\mds{T}$.
Therefore, $\{\mds{T}_t\}_{0\le t\le 1}$ is the desired continuous path connecting $\mds{1}$ and $\mds{T}$ and satisfying $\mds{T}_t\vb{1}/d=\vb{r}_t$ for all $t\in[0,1]$.
\end{proof}

\begin{proposition}\label{prop:scale.sol}
Let $\mds{A}=[A_{mn}]\in\mbb{R}_{>0}^{d\times d}$ be a stochastic map, and let $\vb{r}\in\mca{P}_d^+$ be a positive distribution.
Then there exists a unique pair of vectors $\vb{u}$ and $\vb{v}$ such that $\vb{1}^\top\vb{v}=0$ and the matrix $\mds{T}=\diag(e^{\vb{u}})\mds{A}\diag(e^{\vb{v}})$ satisfies $\mds{T}\vb{1}/d=\vb{r}$ and $\mds{T}^\top\vb{1}=\vb{1}$.
\end{proposition}
\begin{proof}
To establish the unique existence of $(\vb{u},\vb{v})$, consider the function
\begin{equation}
	F(\vb{u},\vb{v})=\sum_{m,n}e^{u_m+v_n}A_{mn} - d\sum_m r_mu_m - \sum_nv_n.
\end{equation}
Since
\begin{align}
	\pdv{F(\vb{u},\vb{v})}{u_m}&=\sum_{n}e^{u_m+v_n}A_{mn}-dr_m,\\
	\pdv{F(\vb{u},\vb{v})}{v_n}&=\sum_{m}e^{u_m+v_n}A_{mn}-1,
\end{align}
a global minimum $(\vb{u},\vb{v})$ of $F$ must satisfy the stationarity conditions $\pdv*{F}{u_m}=0$ and $\pdv*{F}{v_n}=0$, which are exactly equivalent to $\mds{T}\vb{1}/d=\vb{r}$ and $\mds{T}^\top\vb{1}=\vb{1}$.
Therefore, it suffices to prove the uniqueness of the global minimum of $F(\vb{u},\vb{v})$.
Since the Hessian matrix of $F$ has the form
\begin{equation}
	\mds{H}=\begin{pmatrix}
		\diag(\mds{T}\vb{1}) & \mds{T}\\
		\mds{T}^\top & \diag(\mds{T}^\top\vb{1})
	\end{pmatrix},
\end{equation}
its quadratic form $\vb{z}^\top\mds{H}\vb{z}$ reads
\begin{equation}
	\vb{z}^\top\mds{H}\vb{z}=\sum_{m,n}T_{mn}(x_m+y_n)^2\ge 0,
\end{equation}
where $\vb{z} = [ x_1, \dots, x_d, y_1, \dots, y_d ]^\top$.
Because $A_{mn}>0$ and $T_{mn}>0$ for all $m$ and $n$, $F$ is strictly convex on the subspace orthogonal to $[\vb{1}^\top,-\vb{1}^\top]^\top$.
Moreover, $F(\vb{u},\vb{v})\to\infty$ as $|u_m|\to\infty$ or $|v_n|\to\infty$.
Hence, $F(\vb{u},\vb{v})$ is continuous, strictly convex, and coercive on the restricted subspace $\vb{1}^\top\vb{v}=0$.
Consequently, it has a unique global minimum $(\vb{u}^*,\vb{v}^*)$, which completes the proof.
\end{proof}

\begin{proposition}\label{prop:equiv.map}
Let $\mds{A}=[A_{mn}]\in\mbb{R}_{\ge 0}^{d\times d}$ be a stochastic map.
Suppose there exist pairs $(\vb{u},\vb{v})$ and $(\widetilde{\vb{u}},\widetilde{\vb{v}})$ such that the two double-scaled matrices $\mds{T}=\diag(e^{\vb{u}})\mds{A}\diag(e^{\vb{v}})$ and $\widetilde{\mds{T}}=\diag(e^{\widetilde{\vb{u}}})\mds{A}\diag(e^{\widetilde{\vb{v}}})$ satisfy $\mds{T}\vb{1}=\widetilde{\mds{T}}\vb{1}$ and $\mds{T}^\top\vb{1}=\widetilde{\mds{T}}^\top\vb{1}=\vb{1}$.
Then $\mds{T}=\widetilde{\mds{T}}$.
\end{proposition}
\begin{proof}
We consider a bipartite graph $G=(U,V,E)$, where $U=\{1,\dots,d\}$ denotes the set of row nodes, $V=\{1,\dots,d\}$ denotes the set of column nodes, and $E=\{(m,n)\in U\times V\,|\,A_{mn}>0\}$ denotes the set of edges connecting $U$ and $V$ whenever $A_{mn}$ is positive.
This graph can always be decomposed into connected components, $G=\bigcup_k G_k$, where $G_k=(U_k,V_k,E_k)$.
Note that $\{U_k\}_k$ are disjoint subsets, and the same holds for $\{V_k\}_k$ and $\{G_k\}_k$.
If $(m,n)\notin E$, then $A_{mn}=0$, which immediately implies $T_{mn}=\widetilde{T}_{mn}$.
Thus, to show $\mds{T}=\widetilde{\mds{T}}$, it suffices to prove $T_{mn}=\widetilde{T}_{mn}$ for $(m,n)\in E$.

To this end, consider each component $G_k$.
If $E_k=\emptyset$, then $A_{mn}=0$ for any $m\in U_k$ or $n\in V_k$, which immediately implies $T_{mn}=\widetilde{T}_{mn}$.
Therefore, we need only consider components $G_k$ with nonempty $E_k$.
Define $x_m=\widetilde{u}_m-u_m$ and $y_n=\widetilde{v}_n-v_n$ for $m\in U_k$ and $n\in V_k$.
Let $\hat{m}=\argmax_{m\in U_k}x_m$ and $\hat{n}=\argmin_{n\in V_k}y_n$.
From $\sum_nT_{\hat{m}n}=\sum_n\widetilde{T}_{\hat{m}n}$ and
$\widetilde{T}_{mn}=e^{\widetilde{u}_m+\widetilde{v}_n}A_{mn}=e^{x_m+y_n}e^{u_m+v_n}A_{mn}=e^{x_m+y_n}T_{mn}$ for all $m,n$, we have
\begin{equation}\label{eq:tmp1}
	\sum_{n\in V_k}T_{\hat{m}n}=\sum_{n\in V_k}T_{\hat{m}n}e^{x_{\hat{m}}+y_n}\Leftrightarrow \sum_{n\in V_k}T_{\hat{m}n}(e^{x_{\hat{m}}+y_n}-1)=0.
\end{equation}
Thus, $x_{\hat{m}}+y_n\le 0$ holds for some $n\in V_k$ (otherwise, the left-hand side of Eq.~\eqref{eq:tmp1} would be strictly positive); hence $x_{\hat{m}}+y_{\hat{n}}\le 0$.
Similarly, from $\sum_mT_{m\hat{n}}=\sum_m\widetilde{T}_{m\hat{n}}$, we have
\begin{equation}\label{eq:tmp2}
	\sum_{m\in U_k}T_{m\hat{n}}=\sum_{m\in U_k}T_{m\hat{n}}e^{x_m+y_{\hat{n}}}\Leftrightarrow \sum_{m\in U_k}T_{m\hat{n}}(e^{x_m+y_{\hat{n}}}-1)=0.
\end{equation}
Thus, $x_m+y_{\hat{n}}\ge 0$ holds for some $m\in U_k$ (otherwise, the left-hand side of Eq.~\eqref{eq:tmp2} would be strictly negative); hence $x_{\hat{m}}+y_{\hat{n}}\ge 0$.
Consequently, $x_{\hat{m}}+y_{\hat{n}}=0$.
It follows that
\begin{equation}
	x_{\hat{m}}+y_n\ge 0=x_{\hat{m}}+y_{\hat{n}}\ge x_m+y_{\hat{n}}.
\end{equation}
This implies that $x_{\hat{m}}+y_n=0$ whenever $T_{\hat{m}n}>0$, and $x_m+y_{\hat{n}}=0$ whenever $T_{m\hat{n}}>0$.
In other words, $y_n=-x_{\hat{m}}=y_{\hat{n}}$ whenever $(\hat{m},n)\in E_k$, and $x_m=-y_{\hat{n}}=x_{\hat{m}}$ whenever $(m,\hat{n})\in E_k$.
By propagating this relation to all nodes in $G_k$, we obtain $x_m=x_{\hat{m}}$ for any $m\in U_k$ and $y_n=y_{\hat{n}}$ for any $n\in V_k$.
Since $x_{\hat{m}}+y_{\hat{n}}=0$, we have, for any $(m,n)\in E_k$,
\begin{equation}
	\widetilde{T}_{mn}=e^{x_m+y_n}T_{mn}=e^{x_{\hat{m}}+y_{\hat{n}}}T_{mn}=T_{mn}.
\end{equation}
This completes the proof.
\end{proof}

\subsection{Calculation of the geometric complexity of the two-level reset map}
The two-level reset map $\mds{T}$ is explicitly given by
\begin{equation}
	\mds{T}=\begin{pmatrix}
		1 & 1-e^{-w\tau}\\
		0 & e^{-w\tau}
	\end{pmatrix}.
\end{equation}
To evaluate the geometric complexity $\mca{C}(\mds{T})$, it is convenient to consider the curve (C), $e^x+e^y=1$, on the two-dimensional $xy$ plane.
The vectors $\vb{r}^{\mds{1}}=[1/2,1/2]^\top$ and $\vb{r}^{\mds{T}}=[1-e^{-w\tau}/2,e^{-w\tau}/2]^\top$ are mapped to this plane as the points $P(x_P,y_P)$ and $Q(x_Q,y_Q)$, where $x_P=y_P=\ln(1/2)$, $x_Q=\ln(1-e^{-w\tau}/2)$, and $y_Q=\ln(e^{-w\tau}/2)$.
Then $\mca{C}(\mds{T})$ is exactly the arc length between points $P$ and $Q$ along curve (C).
From $y=\ln(1-e^x)$, the arc length between $P$ and $Q$ is
\begin{equation}
	L=\int_{x_P}^{x_Q}\dd{x}\sqrt{1+\qty(\dv{y}{x})^2}=\int_{x_P}^{x_Q}\dd{x}\frac{\sqrt{2e^{2x}-2e^x+1}}{1-e^x}.
\end{equation}
Setting $z=2e^x-1$, the integral becomes
\begin{equation}
	\int\dd{x}\frac{\sqrt{2e^{2x}-2e^x+1}}{1-e^x}=\sqrt{2}\int\dd{z}\frac{\sqrt{1+z^2}}{1-z^2}.
\end{equation}
By direct integration, we obtain
\begin{equation}
	\int\dd{z}\frac{\sqrt{1+z^2}}{1-z^2}=\sqrt{2}\atanh\frac{\sqrt{2}z}{\sqrt{1+z^2}}-\atanh\frac{z}{\sqrt{1+z^2}}+\text{const}.
\end{equation}
Since $z=0$ at $x=x_P$ and $z=1-e^{-w\tau}$ at $x=x_Q$, the arc length is
\begin{align}
	L&=2\atanh[\sqrt{2}f(w\tau)]-\sqrt{2}\atanh[f(w\tau)],
\end{align}
where $f(u)=(1-e^{-u})/\sqrt{1+(1-e^{-u})^2}$.
Therefore, the complexity is $\mca{C}(\mds{T})=L$.

\subsubsection{Asymptotic expression of the complexity}
We evaluate $\mca{C}(\mds{T})$ in the regime $w\tau\gg 1$.
Let $z=e^{-w\tau}\ll 1$. The complexity is then evaluated as follows:
\begin{align}
	\mca{C}(\mds{T})&=2\atanh\qty(\frac{\sqrt{2}(1-z)}{\sqrt{1+(1-z)^2}})-\sqrt{2}\atanh\qty(\frac{1-z}{\sqrt{1+(1-z)^2}})\notag\\
	&=\ln\frac{4-z+O(z^2)}{z+O(z^2)}-\frac{1}{\sqrt{2}}\ln\frac{2+2\sqrt{2}-z+O(z^2)}{2\sqrt{2}-2+z+O(z^2)}\notag\\
	&=-\ln z + \ln 4 - \frac{1}{\sqrt{2}}\ln \frac{1+\sqrt{2}}{\sqrt{2}-1} + O(z)\notag\\
	&=w\tau + 2\ln 2 - \sqrt{2}\ln(\sqrt{2}+1) + O(e^{-w\tau}).
\end{align}

\subsubsection{Analytical proof of $e^{\mca{C}(\mds{T})}\times\epsilon(\mds{T})\ge 1$}
Since $\epsilon(\mds{T})=e^{-w\tau}$, it suffices to prove
\begin{equation}\label{eq:ana.proof}
	\mca{C}(\mds{T})\ge w\tau.
\end{equation}
Let $g(u)=2\atanh[\sqrt{2}f(u)]-\sqrt{2}\atanh[f(u)]$, then $\mca{C}(\mds{T})=g(w\tau)$.
Setting $z=1-e^{-u}\ge 0$, we obtain $f(u)=z/\sqrt{1+z^2}$.
Differentiating $g(u)$ with respect to $u$ gives
\begin{align}
	\dv{g}{u}&=\dv{g}{f}\dv{f}{z}\dv{z}{u}\notag\\
	&=\sqrt{2}\frac{(1+z^2)^2}{1-z^2}\frac{1}{(1+z^2)^{3/2}}(1-z)\notag\\
	&=\frac{\sqrt{2(1+z^2)}}{1+z}\ge 1.
\end{align}
Therefore, $g(u)\ge g(0)+u=u$, which immediately proves Eq.~\eqref{eq:ana.proof}.

\subsection{Generalization to the continuous-state case}
Here, we generalize the complexity theory to the continuous-variable setting.
For simplicity, we consider stochastic maps acting on probability distributions over the one-dimensional interval $[-1,1]$.
A stochastic map $\mds{T}$ is defined via a transition kernel $\mca{K}(x,y)$ as
\begin{equation}
	(\mds{T}p)(x)=\int_{-1}^1\dd{y}\mca{K}(x,y)p(y),
\end{equation}
where $\mca{K}(x,y)$ satisfies the conditions $\mca{K}(x,y)\ge 0$ and $\int_{-1}^1\dd{x}\mca{K}(x,y)=1$.
As in the discrete-state case, for each stochastic map $\mds{T}$, we define $r^{\mds{T}}\coloneqq\mds{T}\bar{p}$, where $\bar{p}(x)=1/2$ denotes the uniform distribution.
In other words, $r^{\mds{T}}$ is the output distribution obtained by applying the map $\mds{T}$ to the uniform distribution.
For any reset map $\mds{T}$, we define its geometric complexity as
\begin{equation}
	\mca{C}(\mds{T})\coloneqq\min_{\{\mds{T}_t\}}\int_0^1\dd{t}\sqrt{g_{p_t}(\dot p_t,\dot p_t)},
\end{equation}
where $p_t=r^{\mds{T}_t}$, and the minimum is taken over all continuous paths $\{\mds{T}_t\}_{0\le t\le 1}$ of maps connecting the identity map $\mds{1}$ to the given map $\mds{T}$.
Here, the identity map $\mds{1}$ is associated with the kernel $\mca{K}(x,y)=\delta(x-y)$.
The Riemannian metric $g_p$ is defined analogously to that in the discrete-state case:
\begin{equation}
	g_{p}(u,v)=\int_{-1}^1\dd{x}\frac{u(x)v(x)}{p(x)^2}.
\end{equation}

\subsubsection{Protocol-scaling relation}
Consider a reset map $\mds{T}$ constructed by sequentially applying $N$ protocols, where the $k$th protocol is realized by an overdamped Langevin process with a fixed potential $U_k(x)$ over a unit time interval.
We impose the physically relevant constraint that the maximum height of the potential $U_k(x)$ remains finite at all times, namely, $0\le U_k(x)\le U_m$ for any $-1<x<1$ and $1\le k\le N$, where $U_m$ is a positive constant.
Let $\mca{F}_k$ denote the Fokker--Planck generator corresponding to the $k$th protocol, defined by
\begin{equation}
	\mca{F}_k(p)=\pdv{x}(\pdv{U_k}{x}p + D\pdv{x}p).
\end{equation}
Then, the map $\mds{T}$ is determined by the kernel
\begin{equation}
	\mca{K}(x,y)=e^{\mca{F}_N}\dots e^{\mca{F}_1}[\delta(x-y)].
\end{equation}
In what follows, we show how the complexity of this map is related to the number of protocols.

First, let $p_k(x)=e^{\mca{F}_k}[p_{k-1}(x)]$, where $p_0=\bar{p}$ is the uniform distribution.
Define the equilibrium distribution associated with the potential $U_k(x)$ by $\pi_k(x)=e^{-\beta U_k(x)}/Z_k$.
We prove by induction that
\begin{equation}
	e^{-2\beta kU_m}\pi_k(x)\le p_k(x)\le e^{2\beta kU_m}\pi_k(x),
\end{equation}
where $U_0(x)=0$.
This is evident for $k=0$.
Assume that $e^{-2\beta kU_m}\pi_k(x)\le p_k(x)\le e^{2\beta kU_m}\pi_k(x)$ holds for some $k\ge 0$.
Since $0\le U_{k}(x),U_{k+1}(x)\le U_m$, we have
\begin{align}
	\pi_k(x)&=\frac{e^{-\beta U_k(x)}}{\int_{-1}^1\dd{x}e^{-\beta U_k(x)}}\notag\\
	&\ge e^{-2\beta U_m}\frac{e^{-\beta U_{k+1}(x)}}{\int_{-1}^1\dd{x}e^{-\beta U_{k+1}(x)}}\notag\\
	&=e^{-2\beta U_m}\pi_{k+1}(x)
\end{align}
and
\begin{align}
	\pi_k(x)&\le e^{2\beta U_m}\frac{e^{-\beta U_{k+1}(x)}}{\int_{-1}^1\dd{x}e^{-\beta U_{k+1}(x)}}\notag\\
	&=e^{2\beta U_m}\pi_{k+1}(x).
\end{align}
Therefore,
\begin{equation}
	e^{-2\beta(k+1)U_m}\pi_{k+1}(x)\le p_{k}(x)\le e^{2\beta(k+1)U_m}\pi_{k+1}(x).
\end{equation}
It then follows immediately from Proposition \ref{prop:pdf.lb} that $p_{k+1}(x)=e^{\mca{F}_{k+1}}[p_k(x)]$ satisfies
\begin{equation}
	e^{-2\beta(k+1)U_m}\pi_{k+1}(x)\le p_{k+1}(x)\le e^{2\beta(k+1)U_m}\pi_{k+1}(x).
\end{equation}

We are now ready to establish the connection between the number of protocols and the geometric complexity by deriving a lower bound on $N$ in terms of $\mca{C}(\mds{T})$.
To this end, consider a smooth path of stochastic maps $\{\mds{T}_t\}_{0\le t\le 1}$ connecting $\mds{1}$ and $\mds{T}$ such that
\begin{equation}
	\ln r^{\mds{T}_t}(x)=\ln r^{\mds{1}}(x) + [\ln r^{\mds{T}}(x) - \ln r^{\mds{1}}(x)]t + c_t,
\end{equation}
where $c_t$ is a constant chosen to satisfy the normalization condition $\int_{-1}^1\dd{x}r^{\mds{T}_t}(x)=1$.
Such a path always exists by Proposition \ref{prop:ovd.path.exist}.
Defining $q(x)\coloneqq\ln r^{\mds{T}}(x) - \ln r^{\mds{1}}(x)=\ln[2r^{\mds{T}}(x)]$, we can write $c_t$ as
\begin{equation}
	c_t=-\ln\qty(\int_{-1}^1\dd{x}e^{-\ln 2 + q(x)t}).
\end{equation}
The time derivative of $c_t$ is given by
\begin{align}
	\dot c_t&=-\frac{\int_{-1}^1\dd{x}q(x)e^{q(x)t}}{\int_{-1}^1\dd{x}e^{q(x)t}}.
\end{align}
Applying the Cauchy--Schwarz inequality, we obtain
\begin{align}
	&\int_{-1}^1\dd{x}q(x)^2e^{q(x)t}\int_{-1}^1\dd{x}e^{q(x)t}\notag\\
	&\ge \qty[\int_{-1}^1\dd{x}q(x)e^{q(x)t}]^2,
\end{align}
which implies that $\ddot{c}_t\le 0$.
Therefore, $|\dot{c}_t|\le\max\qty(|\dot{c}_0|,|\dot{c}_1|)$.
Since $p_N(x)=(\mds{T}\bar{p})(x)=r^{\mds{T}}(x)$ satisfies
\begin{equation}
	e^{-2\beta NU_m}\pi_{N}(x)\le p_{N}(x)\le e^{2\beta NU_m}\pi_{N}(x),
\end{equation}
it follows that
\begin{equation}
	\frac{1}{2}e^{-\beta(2N+1)U_m}\le r^{\mds{T}}(x)\le \frac{1}{2}e^{\beta(2N+1)U_m}.
\end{equation}
Using this inequality, we obtain
\begin{align}
	|\dot{c}_0|&=\frac{1}{2}\qty|\int_{-1}^1\dd{x}[\ln r^{\mds{T}}(x) - \ln r^{\mds{1}}(x)]|\notag\\
	&\le \beta(2N+1)U_m,\\
	|\dot{c}_1|&=\qty|\int_{-1}^1\dd{x}[\ln r^{\mds{T}}(x) - \ln r^{\mds{1}}(x)]r^{\mds{T}}(x)|\notag\\
	&\le \beta(2N+1)U_m.
\end{align}
In other words, $|\dot{c}_t|\le \beta(2N+1)U_m$.
Combining this with the bound $|\ln r^{\mds{T}}(x) - \ln r^{\mds{1}}(x)|\le \beta(2N+1)U_m$, we obtain the following upper bound on the complexity:
\begin{align}
	\mca{C}(\mds{T})&\le \int_0^1\dd{t}\sqrt{\int_{-1}^1\dd{x}\qty(\dv{\ln r^{\mds{T}_t}(x)}{t})^2}\notag\\
	&=\int_0^1\dd{t}\sqrt{\int_{-1}^1\dd{x}(\ln r^{\mds{T}}(x) - \ln r^{\mds{1}}(x)+\dot{c}_t)^2}\notag\\
	&\le 2\sqrt{2}\beta(2N+1)U_m.
\end{align}
Consequently, the number of protocols is lower bounded in terms of the complexity as
\begin{equation}
	N\ge \frac{\mca{C}(\mds{T})}{4\sqrt{2}\beta U_m}-\frac{1}{2}.
\end{equation}
This relation also indicates that the geometric complexity captures the maximum height of the potential, since perfect resetting can be achieved in a finite number of operations provided that $U_m$ diverges.

\subsubsection{The complexity-error trade-off relation}
The error of the reset map $\mds{T}$ associated with an undesired region $\mfr{u}=[0,1]$ is defined as
\begin{equation}
	\epsilon(\mds{T})\coloneqq\int_{\mfr{u}}\dd{x}(\mds{T}\vb{1})(x),
\end{equation}
where $\vb{1}(x)=1$ for all $x\in[-1,1]$.
This error quantifies how effectively the reset map suppresses the probability weight in the undesired region when acting on the maximally mixed distribution.
First, the geometric complexity is lower bounded by
\begin{equation}
	\mca{C}(\mds{T})\ge\sqrt{\int_{-1}^1\dd{x}|\ln r^{\mds{T}}(x) - \ln r^{\mds{1}}(x)|^2}\eqqcolon\ell.
\end{equation}
Since $e^{x}\ge e^{-|x|}$ for any real $x$ and $e^{-x}$ is convex, we obtain the lower bound
\begin{align}
	\epsilon(\mds{T})&=\int_{\mfr{u}}\dd{x}e^{\ln[2r^{\mds{T}}(x)]}\notag\\
	&\ge \int_{\mfr{u}}\dd{x}e^{- |\ln[2r^{\mds{T}}(x)]|}\notag\\
	&\ge e^{-\int_{\mfr{u}}\dd{x}|\ln[2r^{\mds{T}}(x)]|},\label{eq:ovd.error.lb1}
\end{align}
where the last inequality follows from Jensen's inequality.
Next, by the Cauchy--Schwarz inequality, $\ell$ is lower bounded as
\begin{equation}\label{eq:ovd.ell.lb1}
	\ell\ge\int_{\mfr{u}}\dd{x}|\ln[2r^{\mds{T}}(x)]|.
\end{equation}
Combining Eqs.~\eqref{eq:ovd.error.lb1} and \eqref{eq:ovd.ell.lb1}, we obtain
\begin{equation}
	\epsilon(\mds{T})\ge e^{-\ell}.
\end{equation}
Finally, since $\ell\le\mca{C}(\mds{T})$, we immediately obtain the desired relation \eqref{eq:geo.comp.3rd.law}:
\begin{equation}
	e^{\mca{C}(\mds{T})}\times\epsilon(\mds{T})\ge 1.
\end{equation}

\begin{proposition}\label{prop:pdf.lb}
Consider a particle confined by a time-independent potential $U(x)$ and governed by overdamped Langevin dynamics.
Let $p_0(x)$ be an initial probability distribution satisfying $c_1\pi(x)\le p_0(x)\le c_2\pi(x)$, where $\pi(x)=e^{-\beta U(x)}/Z$ is the equilibrium distribution and $c_1,c_2$ are positive constants.
Then $c_1\pi(x)\le p_t(x)\le c_2\pi(x)$ for all times.
\end{proposition}
\begin{proof}
Let $q_t(x)\coloneqq p_t(x)-c_1\pi(x)$ and $s_t(x)\coloneqq c_2\pi(x)-p_t(x)$.
Then $q_0(x)\ge 0$ and $s_0(x)\ge 0$ for any $x\in[-1,1]$.
Let $\mca{F}$ denote the Fokker--Planck generator. Since
\begin{align}
	\dv{t}q_t&=\mca{F}(p_t)=\mca{F}(q_t) - c_1\mca{F}(\pi)=\mca{F}(q_t),\\
	\dv{t}s_t&=-\mca{F}(p_t)=\mca{F}(s_t) - c_2\mca{F}(\pi)=\mca{F}(s_t),
\end{align}
and the Fokker--Planck generator $\mca{F}$ is positivity-preserving, the initial conditions $q_0(x)\ge 0$ and $s_0(x)\ge 0$ imply that $q_t(x)\ge 0$ and $s_t(x)\ge 0$ for any $t$.
Therefore, $c_1\pi(x)\le p_t(x)\le c_2\pi(x)$.
\end{proof}

\begin{proposition}\label{prop:ovd.path.exist}
For any smooth path $\{p_t\}_{0\le t\le 1}$ of strictly positive probability distributions connecting $r^{\mds{1}}$ and $r^{\mds{T}}$, there exists a corresponding continuous path $\{\mds{T}_t\}_{0\le t\le 1}$ of stochastic maps connecting $\mds{1}$ and $\mds{T}$ such that $r^{\mds{T}_t}=p_t$ for all $t\in[0,1]$.
\end{proposition}
\begin{proof}
The proof strategy is similar to that in the discrete-state case.
Let $\mca{K}$ and $\delta$ denote the transition kernels associated with the map $\mds{T}$ and the identity map $\mds{1}$, respectively.
We consider a smooth path of stochastic maps $\{\mds{A}_{t,\kappa}\}_{0\le t\le 1}$ with kernels $\{\mca{A}_{t,\kappa}\}_{0\le t\le 1}$ defined by
\begin{equation}
	\mca{A}_{t,\kappa} = [\kappa + a\sin(\pi t)]\mca{I} + [1- \kappa - a\sin(\pi t)][(1-t)\delta + t\mca{K}],
\end{equation}
where $\mca{I}(x,y)=1/2$ for all $x,y\in[-1,1]$, and $0<a<1$ and $0<\kappa<1-a$ are arbitrary constants.
Clearly, each kernel $\mca{A}_{t,\kappa}$ defines a valid stochastic map $\mds{A}_{t,\kappa}$, $\lim_{\kappa\to 0}\mca{A}_{0,\kappa}=\delta$, $\lim_{\kappa\to 0}\mca{A}_{1,\kappa}=\mca{K}$, and all transition probabilities $\mca{A}_{t,\kappa}(x,y)$ are strictly positive for $x,y\in[-1,1]$ and $t\in[0,1]$.
However, this path of stochastic maps does not satisfy the required condition $r^{\mds{A}_{t,\kappa}}=p_t$ for all $t\in[0,1]$.
To construct the desired path, we consider a path of stochastic maps determined by the following double-scaled kernels:
\begin{equation}\label{eq:ovd.scale.map}
	\mca{K}_{t,\kappa}(x,y)=e^{u_{t,\kappa}(y)+v_{t,\kappa}(x)}\mca{A}_{t,\kappa}(x,y),
\end{equation}
where $u_{t,\kappa}(x)$ and $v_{t,\kappa}(x)$ are real-valued functions on $[-1,1]$.
According to Proposition \ref{prop:ovd.scale.sol}, there exists a path $\{(u_{t,\kappa},v_{t,\kappa})\}_{0\le t\le 1}$ such that $\int_{-1}^1\dd{x}\mca{K}_{t,\kappa}(x,y)=1$, $r^{\mds{T}_{t,\kappa}}=\kappa\bar{p} + (1-\kappa) p_t$, and $\int_{-1}^1\dd{x}v_{t,\kappa}(x)=0$, where $\mds{T}_{t,\kappa}$ denotes the stochastic map associated with the kernel $\mca{K}_{t,\kappa}$.
The gauge condition $\int_{-1}^1\dd{x}v_{t,\kappa}(x)=0$ removes the translational invariance of the log-linear scaling; that is, $u_{t,\kappa}\leftarrow u_{t,\kappa}+c$ and $v_{t,\kappa}\leftarrow v_{t,\kappa}-c$ yield the same map $\mds{T}_{t,\kappa}$ for any real $c$.
By the implicit function theorem, the path $\{(u_{t,\kappa},v_{t,\kappa})\}_{0\le t\le 1}$ is smooth. Consequently, the corresponding path of stochastic maps $\{\mds{T}_{t,\kappa}\}_{0\le t\le 1}$ is continuous.

Next, we show that $\mca{K}_{0,\kappa}=\mca{A}_{0,\kappa}$ and $\mca{K}_{1,\kappa}=\mca{A}_{1,\kappa}$.
Since
\begin{align}
	\mca{A}_{0,\kappa}&=e^{0+0}\mca{A}_{0,\kappa},~r^{\mds{A}_{0,\kappa}}=r^{\mds{T}_{0,\kappa}},\\
	\mca{A}_{1,\kappa}&=e^{0+0}\mca{A}_{1,\kappa},~r^{\mds{A}_{1,\kappa}}=r^{\mds{T}_{1,\kappa}},
\end{align}
the uniqueness of $(u,v)$ in Proposition \ref{prop:ovd.scale.sol} implies that $u_{0,\kappa}=u_{1,\kappa}=0$ and $v_{0,\kappa}=v_{1,\kappa}=0$.
Therefore, the path $\{\mds{T}_{t,\kappa}\}_{0\le t\le 1}$ smoothly connects $\mds{A}_{0,\kappa}$ and $\mds{A}_{1,\kappa}$.
We now take the limit $\kappa\to 0$ and define $\mds{T}_t\coloneqq\lim_{\kappa\to 0}\mds{T}_{t,\kappa}$.
The resulting path $\{\mds{T}_t\}_{0\le t\le 1}$ is continuous, connects $\mds{1}$ and $\mds{T}$, and satisfies $r^{\mds{T}_t}=p_t$ for all $t\in[0,1]$.
This completes the proof.
\end{proof}

\begin{proposition}\label{prop:ovd.scale.sol}
Let $\mca{A}$ be a strictly positive transition kernel and let $p$ be a positive probability distribution.
Then there exists a unique pair of real-valued functions $u(x)$ and $v(x)$ such that $\int_{-1}^1\dd{x}v(x)=0$ and $\mca{K}(x,y)=e^{u(y)+v(x)}\mca{A}(x,y)$ is a transition kernel satisfying $r^{\mds{T}}=p$, where $\mds{T}$ denotes the stochastic map associated with the kernel $\mca{K}$.
\end{proposition}
\begin{proof}
Let $\mca{K}(x,y)=e^{u(y)+v(x)}\mca{A}(x,y)$ and consider the following functional:
\begin{align}
	F(u,v)&=\int_{-1}^1\dd{x}\int_{-1}^1\dd{y}e^{u(y)+v(x)}\mca{A}(x,y)\notag\\
	&-\int_{-1}^1\dd{y}u(y)-2\int_{-1}^1\dd{x}v(x)p(x).
\end{align}
It is straightforward to verify that $F(u,v)$ is a convex functional.
Since
\begin{align}
	\frac{\delta F}{\delta u(y)}&=\int_{-1}^1\dd{x}\mca{K}(x,y)-1,\\
	\frac{\delta F}{\delta v(x)}&=\int_{-1}^1\dd{y}\mca{K}(x,y)-2p(x),
\end{align}
the global minimum of $F(u,v)$ is achieved only when ${\delta F}/{\delta u(y)}=0$ and ${\delta F}/{\delta v(x)}=0$, which are exactly the required conditions $\int_{-1}^1\dd{x}\mca{K}(x,y)=1$ and $r^{\mds{T}}=p$.
Let $f(t)=F(u+tu',v+tv')$, where $u'$ and $v'$ are arbitrary real-valued functions.
The second derivative of $f(t)$ with respect to $t$ is
\begin{align}
	\eval{\dv[2]{f(t)}{t}}_{t=0}&=\int_{-1}^1\dd{x}\int_{-1}^1\dd{y}[u'(y)+v'(x)]^2\mca{K}(x,y)\ge 0.
\end{align}
The equality holds only when $u'(y)+v'(x)=0$ for all $x,y\in[-1,1]$.
If we restrict to the subspace where $\int_{-1}^1\dd{x}v(x)=0$, then $u'(y)=v'(x)=0$.
Therefore, $F$ is strictly convex on the restricted subspace.
Furthermore, since $F$ is continuous and coercive, it admits a unique global minimum, which completes the proof.
\end{proof}

\section{Detailed derivations for the quantum case}\label{app:qua.details}

For convenience, we introduce the following notation.
Let $\mca{S}_d$ denote the set of density operators on the $d$-dimensional complex Hilbert space $\mca{H}_d$, and let $\Upsilon_d$ denote the space of quantum channels $\Lambda:\mca{S}_d\mapsto\mca{S}_d$.
Specifically,
\begin{align}
	\mca{S}_d&=\{\rho\in\mca{H}_d\,|\,\rho\succeq 0~\text{and}~\tr\rho=1\},\\
	\Upsilon_d&=\{\Lambda\,|\,\Lambda(\rho)\succeq 0~\text{and}~\tr[\Lambda(\rho)]=1~\forall\rho\in\mca{S}_d\}.
\end{align}
In addition, let $\mca{S}_d^+=\{\rho\in\mca{S}_d\,|\,\rho\succ 0\}$ denote the subset of positive-definite density operators.
The identity quantum channel is denoted by $\Id$ [i.e., $\Id(\circ)=\circ$].

We use the Choi-matrix representation of quantum channels.
For each quantum channel $\Lambda$, the corresponding Choi matrix is defined by
\begin{equation}
	\mds{M}^\Lambda \coloneqq (\mds{1}\otimes\Lambda)(\dyad{\Omega}),
\end{equation}
where $\ket{\Omega}=\sum_{n=1}^d\ket{n}\otimes\ket{n}$ and $\{\ket{n}\}$ is an orthonormal basis of $\mca{H}_d$.
The Choi matrix can be written as $\mds{M}^\Lambda=\sum_{m,n}\dyad{m}{n}\otimes\Lambda(\dyad{m}{n})$, and the channel $\Lambda$ is related to it via
\begin{equation}
	\Lambda(\rho)=\tr_1[(\rho^\top\otimes\mds{1})\mds{M}^\Lambda],
\end{equation}
where $\rho^\top=\sum_{m,n}\mel{n}{\rho}{m}\dyad{m}{n}$.
Since $\Lambda$ is CPTP, the Choi matrix is positive semidefinite and satisfies $\tr_2\mds{M}^\Lambda=\mds{1}$.
By the Choi--Jamiolkowski isomorphism, there is a one-to-one correspondence between Choi matrices and quantum channels (CPTP maps) \cite{Choi.1975.LAA,Jamiokowski.1972.RMP}.
It also follows that $\phi_{\mds{M}^\Lambda}=d^{-1}\tr_1\mds{M}^\Lambda=\Lambda(\mds{1}/d)$ is a density operator.

\subsection{Geodesic distance of the defined Riemannian metric in the unconstrained case}\label{app:geo.dis.pos.Herm}

For each Choi matrix $\mds{M}$ representing a quantum channel, the corresponding reduced operator $\phi_{\mds{M}}=d^{-1}\tr_1\mds{M}$ is a density operator.
Therefore, the metric $g_{\mds{M}}(\mds{X},\mds{Y})$ on the space of Choi matrices is equivalent to the metric $g_{\phi_{\mds{M}}}(\phi_{\mds{X}},\phi_{\mds{Y}})=\tr(\phi_{\mds{M}}^{-1}\phi_{\mds{X}}\phi_{\mds{M}}^{-1}\phi_{\mds{Y}})$ on the space of reduced density operators.
In what follows, we compute the geodesic distance induced by the metric $g_{\phi_{\mds{M}}}$ on the space of positive Hermitian operators.
We show that in the \emph{absence} of the trace constraint (i.e., $\tr\phi_{\mds{M}}=1$), this geodesic distance can be obtained analytically.

Let $\mca{D}(\mds{A},\mds{B})$ denote the geodesic distance between two positive Hermitian operators $\mds{A}$ and $\mds{B}$ induced by the above metric.
Note that, in general, neither $\tr\mds{A}$ nor $\tr\mds{B}$ is required to equal $1$.
We prove that $\mca{D}(\mds{A},\mds{B})=\|\ln(\mds{A}^{-1/2}\mds{B}\mds{A}^{-1/2})\|_F$ \cite{Bhatia.2006.LAA} by establishing both $\mca{D}(\mds{A},\mds{B})\ge \|\ln(\mds{A}^{-1/2}\mds{B}\mds{A}^{-1/2})\|_F$ and $\mca{D}(\mds{A},\mds{B})\le \|\ln(\mds{A}^{-1/2}\mds{B}\mds{A}^{-1/2})\|_F$.

We first prove the lower bound.
Let $\{\mds{O}_t\}_{0\le t\le 1}$ be any smooth path of positive Hermitian operators connecting $\mds{A}$ and $\mds{B}$.
Defining $\mds{J}_t\coloneqq\mds{O}_0^{-1/2}\mds{O}_t\mds{O}_0^{-1/2}=e^{\mds{H}_t}$, the metric can be written as
\begin{align}
	g_{\mds{O}_t}(\dot{\mds{O}}_t,\dot{\mds{O}}_t)&=\tr(\mds{O}_t^{-1}\dot{\mds{O}}_t\mds{O}_t^{-1}\dot{\mds{O}}_t)\notag\\
	&=\tr(\mds{J}_t^{-1}\dot{\mds{J}}_t\mds{J}_t^{-1}\dot{\mds{J}}_t)\notag\\
	&=\|e^{-\mds{H}_t/2}\dot{\mds{J}}_te^{-\mds{H}_t/2}\|_F^2.
\end{align}
Using the spectral decomposition $\mds{H}_t=\sum_n h_n\dyad{n}$, where $\{h_n\}$ are real eigenvalues, we obtain
\begin{widetext}
\begin{align}
	\|e^{-\mds{H}_t/2}\dot{\mds{J}}_te^{-\mds{H}_t/2}\|_F^2&=\left\|\int_0^1\dd{s}e^{-\mds{H}_t/2}e^{s\mds{H}_t}\dot{\mds{H}}_te^{(1-s)\mds{H}_t}e^{-\mds{H}_t/2}\right\|_F^2\notag\\
	&=\tr\qty[ \int_0^1\dd{s}\int_0^1\dd{s'} e^{(1/2-s')\mds{H}_t}\dot{\mds{H}}_te^{(s'-1/2)\mds{H}_t} e^{(s-1/2)\mds{H}_t}\dot{\mds{H}}_te^{(1/2-s)\mds{H}_t}]\notag\\
	&=\sum_{m,n}|\mel{m}{\dot{\mds{H}}_t}{n}|^2\int_0^1\dd{s}\int_0^1\dd{s'}e^{(s+s'-1)(h_n-h_m)}\notag\\
	&\ge\sum_{m,n}|\mel{m}{\dot{\mds{H}}_t}{n}|^2\notag\\
	&=\|\dot{\mds{H}}_t\|_F^2.
\end{align}
\end{widetext}
Here, we use Proposition \ref{prop:int.ine} to derive the fourth line.
Therefore, the geodesic distance is lower bounded as
\begin{align}
	\mca{D}(\mds{A},\mds{B})&=\min_{\{\mds{O}_t\}}\int_0^1\dd{t}\sqrt{g_{\mds{O}_t}(\dot{\mds{O}}_t,\dot{\mds{O}}_t)}\notag\\
	&\ge \min_{\{\mds{H}_t\}}\int_0^1\dd{t}\|\dot{\mds{H}}_t\|_F\notag\\
	&= \|\ln\mds{1}-\ln(\mds{A}^{-1/2}\mds{B}\mds{A}^{-1/2})\|_F\notag\\
	&=\|\ln(\mds{A}^{-1/2}\mds{B}\mds{A}^{-1/2})\|_F.\label{eq:D.lb}
\end{align}
Next, we prove the upper bound.
We consider the smooth path $\{\mds{O}_t=\mds{A}^{1/2}(\mds{A}^{-1/2}\mds{B}\mds{A}^{-1/2})^t\mds{A}^{1/2}\}_{0\le t\le 1}$, which connects $\mds{A}$ and $\mds{B}$.
The time derivative of $\mds{O}_t$ is
\begin{align}
	\dot{\mds{O}}_t&=\mds{A}^{1/2}(\mds{A}^{-1/2}\mds{B}\mds{A}^{-1/2})^t\ln(\mds{A}^{-1/2}\mds{B}\mds{A}^{-1/2})\mds{A}^{1/2}\notag\\
	&=\mds{O}_t\mds{A}^{-1/2}\ln(\mds{A}^{-1/2}\mds{B}\mds{A}^{-1/2})\mds{A}^{1/2}.
\end{align}
Therefore, the metric can be written as
\begin{align}
	g_{\mds{O}_t}(\dot{\mds{O}}_t,\dot{\mds{O}}_t)&=\tr(\mds{O}_t^{-1}\dot{\mds{O}}_t\mds{O}_t^{-1}\dot{\mds{O}}_t)\notag\\
	&=\tr[(\ln(\mds{A}^{-1/2}\mds{B}\mds{A}^{-1/2}))^2]\notag\\
	&=\|\ln(\mds{A}^{-1/2}\mds{B}\mds{A}^{-1/2})\|_F^2,
\end{align}
which is constant in time.
Consequently, the geodesic distance is upper bounded as
\begin{align}
	\mca{D}(\mds{A},\mds{B})&\le\int_0^1\dd{t}\sqrt{g_{\mds{O}_t}(\dot{\mds{O}}_t,\dot{\mds{O}}_t)}\notag\\
	&=\|\ln(\mds{A}^{-1/2}\mds{B}\mds{A}^{-1/2})\|_F.\label{eq:D.ub}
\end{align}
Combining Eqs.~\eqref{eq:D.lb} and \eqref{eq:D.ub}, we obtain $\mca{D}(\mds{A},\mds{B})=\|\ln(\mds{A}^{-1/2}\mds{B}\mds{A}^{-1/2})\|_F$.

\begin{proposition}\label{prop:int.ine}
The following inequality holds for any real numbers $x$ and $y$:
\begin{equation}
	I\coloneqq\int_0^1\dd{s}\int_0^1\dd{s'}e^{(s+s'-1)(x-y)}\ge 1.
\end{equation}
\end{proposition}
\begin{proof}
If $x=y$, then $I=1$, so the inequality holds trivially.
Otherwise, letting $z=x-y$, we evaluate the integral as follows:
\begin{align}
	I&=e^{-z}\qty(\int_0^1\dd{s}e^{sz})^2=e^{-z}\qty(\frac{e^{z}-1}{z})^2=\qty[\frac{\sinh(|z|/2)}{|z|/2}]^2\ge 1.
\end{align}
Here, we use $\sinh|z|\ge |z|$ for all $z$ to derive the final inequality.
\end{proof}

\subsection{Error bound for arbitrary initial states}
We show that it suffices to examine the maximally mixed state in order to bound the error of the quantum channel for arbitrary initial states.
Indeed, let $\rho$ be an arbitrary initial state.
Using the linearity of quantum channels, we obtain
\begin{align}
	\epsilon(\Lambda)-\tr\qty[\Pi\Lambda(\rho)]&=\tr\qty[\Pi\Lambda(\mds{1})] - \tr\qty[\Pi\Lambda(\rho)]\notag\\
	&=\tr\qty[\Pi\Lambda(\mds{1}-\rho)].
\end{align}
Since $\Lambda$ is a CPTP map and $\mds{1}-\rho\succeq 0$, we have $\tr\qty[\Pi\Lambda(\mds{1}-\rho)]\ge 0$.
Therefore,
\begin{equation}
	\tr\qty[\Pi\Lambda(\rho)]\le \epsilon(\Lambda).
\end{equation}
This means that the error for any initial state $\rho$ is always upper bounded by $\epsilon(\Lambda)$, which is $d$ times the error for the maximally mixed state.
Hence, if the maximally mixed state is reset with sufficiently small error, other states can also be reset reliably.

\subsection{Proof of the existence of a path of quantum channels}

We prove the following theorem.
\begin{theorem}\label{theo:qmap.exist}
Let $\Lambda$ be a quantum channel such that $\Lambda(\mds{1}/d)$ is full-rank, and let $\{\phi_t\}_{0\le t\le 1}$ be a smooth path in $\mca{S}_d^+$ connecting $\mds{1}/d$ and $\Lambda(\mds{1}/d)$.
Then there exists a continuous path of quantum channels $\{\Lambda_t\}_{0\le t\le 1}$ in $\Upsilon_d$ that connects $\Id$ and $\Lambda$, and satisfies $\Lambda_t(\mds{1}/d)=\phi_t$ for all $t\in[0,1]$.
\end{theorem}
\begin{proof}
For convenience, we use the Choi-matrix representation of quantum channels.
It suffices to prove the existence of a continuous path $\{\mds{M}_t\}_{0\le t\le 1}$ in the space of Choi matrices that connects $\mds{M}^{\Id}$ and $\mds{M}^\Lambda$, and satisfies $\tr_1\mds{M}_t=d\phi_t$ and $\tr_2\mds{M}_t=\mds{1}$ for all $t\in[0,1]$.
Note that $\mds{M}^{\Id}=\dyad{\Omega}$.

To this end, we first consider a smooth path $\{\mds{A}_{t,\kappa}\}_{0\le t\le 1}$, defined by
\begin{align}
	\mds{A}_{t,\kappa}&=[\kappa+a\sin(\pi t)]\frac{\mds{1}\otimes\mds{1}}{d} \notag\\
	&+ [1-\kappa - a\sin(\pi t)][(1-t)\mds{M}^{\Id} + t\mds{M}^\Lambda],
\end{align}
where $0<a<1$ and $0<\kappa<1-a$ are arbitrary constants.
Evidently, $\lim_{\kappa\to 0}\mds{A}_{0,\kappa}=\mds{M}^\Id$ and $\lim_{\kappa\to 0}\mds{A}_{1,\kappa}=\mds{M}^\Lambda$, and each $\mds{A}_{t,\kappa}$ is a Choi matrix corresponding to a quantum channel.
Furthermore, $\mds{A}_{t,\kappa}$ is strictly positive definite for all $t\in[0,1]$.
However, the condition $\tr_1\mds{A}_{t,\kappa}=d\phi_t$ may not hold.
To construct the desired smooth path, we consider the following scaled operators:
\begin{equation}
	\mds{M}_{t,\kappa}=e^{\ln\mds{A}_{t,\kappa} + \mds{U}_{t,\kappa}\otimes\mds{1} + \mds{1}\otimes\mds{V}_{t,\kappa}},
\end{equation}
where $\mds{U}_{t,\kappa}$ and $\mds{V}_{t,\kappa}$ are Hermitian operators.
According to Proposition \ref{prop:scale.sol.qu}, there exists a path $\{(\mds{U}_{t,\kappa},\mds{V}_{t,\kappa})\}_{0\le t\le 1}$ such that $\tr_1\mds{M}_{t,\kappa}=\kappa\mds{1} + (1-\kappa)d\phi_t$, $\tr_2\mds{M}_{t,\kappa}=\mds{1}$, and $\tr\mds{V}_{t,\kappa}=0$.
The gauge condition $\tr\mds{V}_{t,\kappa}=0$ resolves the translational invariance of the log-linear scaling (i.e., $\mds{U}_{t,\kappa}\leftarrow\mds{U}_{t,\kappa}+c\mds{1}$ and $\mds{V}_{t,\kappa}\leftarrow\mds{V}_{t,\kappa}-c\mds{1}$ yield the same operator $\mds{M}_{t,\kappa}$ for any real $c$).
By the implicit function theorem, the path $\{(\mds{U}_{t,\kappa},\mds{V}_{t,\kappa})\}_{0\le t\le 1}$ is smooth; therefore, the corresponding path $\{\mds{M}_{t,\kappa}\}_{0\le t\le 1}$ is also smooth.

Next, we prove that $\mds{M}_{0,\kappa}=\mds{A}_{0,\kappa}$ and $\mds{M}_{1,\kappa}=\mds{A}_{1,\kappa}$.
Since
\begin{widetext}
\begin{align}
	\mds{A}_{0,\kappa}&=e^{\ln\mds{A}_{0,\kappa} + 0\otimes\mds{1} + \mds{1}\otimes 0},~\tr_1\mds{M}_{0,\kappa}=\tr_1\mds{A}_{0,\kappa}=\kappa\mds{1}+(1-\kappa)d\phi_0,~\tr_2\mds{M}_{0,\kappa}=\tr_2\mds{A}_{0,\kappa}=\mds{1},\\
	\mds{A}_{1,\kappa}&=e^{\ln\mds{A}_{1,\kappa} + 0\otimes\mds{1} + \mds{1}\otimes 0},~\tr_1\mds{M}_{1,\kappa}=\tr_1\mds{A}_{1,\kappa}=\kappa\mds{1}+(1-\kappa)d\phi_1,~\tr_2\mds{M}_{1,\kappa}=\tr_2\mds{A}_{1,\kappa}=\mds{1},
\end{align}
\end{widetext}
the uniqueness of $(\mds{U},\mds{V})$ in Proposition \ref{prop:scale.sol.qu} implies that $\mds{U}_{0,\kappa}=\mds{U}_{1,\kappa}=0$ and $\mds{V}_{0,\kappa}=\mds{V}_{1,\kappa}=0$.
Therefore, the path $\{\mds{M}_{t,\kappa}\}_{0\le t\le 1}$ smoothly connects $\mds{A}_{0,\kappa}$ and $\mds{A}_{1,\kappa}$.
Now, we take the limit $\kappa\to 0$ and define $\mds{M}_t\coloneqq\lim_{\kappa\to 0}\mds{M}_{t,\kappa}$.
The path $\{\mds{M}_t\}_{0\le t\le 1}$ is then continuous, connects $\mds{M}^{\Id}$ and $\mds{M}^\Lambda$, and satisfies $\tr_1\mds{M}_t=d\phi_t$ and $\tr_2\mds{M}_t=\mds{1}$ for all $t\in[0,1]$.
By converting each Choi matrix $\mds{M}_t$ into the corresponding quantum channel $\Lambda_t$, we obtain the desired path of quantum channels and complete the proof.
\end{proof}

\begin{proposition}\label{prop:scale.sol.qu}
Let $\mds{A}$ be a strictly positive Hermitian operator and let $\phi$ be a positive density operator.
Then there exists a unique pair of Hermitian operators, $\mds{U}$ and $\mds{V}$, such that $\tr\mds{V}=0$ and the operator $\mds{M}=e^{\ln\mds{A} + \mds{U}\otimes\mds{1} + \mds{1}\otimes\mds{V}}$ satisfies $\tr_1\mds{M}=d\phi$ and $\tr_2\mds{M}=\mds{1}$.
\end{proposition}
\begin{proof}
We consider the following function:
\begin{equation}
	F(\mds{U},\mds{V})=\tr(e^{\ln\mds{A} + \mds{U}\otimes\mds{1} + \mds{1}\otimes\mds{V}}) - \tr(\mds{U}) - d\tr(\mds{V}\phi).
\end{equation}
We show that $F(\mds{U},\mds{V})$ is a convex function.
To this end, consider $f(t)=F(\mds{U}+t\mds{U}',\mds{V}+t\mds{V}')$ as a function of $t$, where $\mds{U}'$ and $\mds{V}'$ are arbitrary Hermitian operators.
It suffices to prove that $\dv*[2]{f(0)}{t}\ge 0$ for any Hermitian operators $\mds{U}$, $\mds{U}'$, $\mds{V}$, and $\mds{V}'$.
Using Duhamel's formula
\begin{equation}
	\dv{t}e^{\mds{X}_t}=\int_0^1\dd{s}e^{s\mds{X}_t}\dot{\mds{X}}_te^{(1-s)\mds{X}_t},
\end{equation}
we obtain
\begin{widetext}
\begin{align}
	\dv{f(t)}{t}&=\dv{t}\qty[\tr(e^{\ln\mds{A} + \mds{W} + t\mds{W}'}) - \tr(\mds{U}+t\mds{U}') - d\tr((\mds{V}+t\mds{V}')\phi)]\notag\\
	&=\int_0^1\dd{s}\tr[e^{s(\ln\mds{A}+\mds{W}+t\mds{W}')}\mds{W}'e^{(1-s)(\ln\mds{A}+\mds{W}+t\mds{W}')}] -\tr(\mds{U}') - d\tr(\mds{V}'\phi)\notag\\
	&=\tr(\mds{W}'e^{\ln\mds{A}+\mds{W}+t\mds{W}'}) -\tr(\mds{U}') - d\tr(\mds{V}'\phi),
\end{align}
\end{widetext}
where $\mds{W}=\mds{U}\otimes\mds{1}+\mds{1}\otimes\mds{V}$ and $\mds{W}'=\mds{U}'\otimes\mds{1}+\mds{1}\otimes\mds{V}'$.
A global minimum must satisfy $\dv*{f(0)}{t}=0$ for any $\mds{U}'$ and $\mds{V}'$.
Since
\begin{align}
	\eval{\dv{f(t)}{t}}_{t=0}&=\tr[\mds{U}'(\tr_2(e^{\ln\mds{A}+\mds{W}}) - \mds{1})] \notag\\
	&+ \tr[\mds{V}'(\tr_1(e^{\ln\mds{A}+\mds{W}}) - d\phi)],
\end{align}
$\dv*{f(0)}{t}=0$ is achieved only when $\tr_1e^{\ln\mds{A}+\mds{W}}=d\phi$ and $\tr_2e^{\ln\mds{A}+\mds{W}}=\mds{1}$, which are exactly the required conditions $\tr_1\mds{M}=d\phi$ and $\tr_2\mds{M}=\mds{1}$.
Similarly, the second derivative of $f(t)$ with respect to $t$ is
\begin{align}
	\eval{\dv[2]{f(t)}{t}}_{t=0}&=\int_0^1\dd{s}\tr[\mds{W}'e^{s(\ln\mds{A}+\mds{W})}\mds{W}'e^{(1-s)(\ln\mds{A}+\mds{W})}].
\end{align}
Since $\ln\mds{A}+\mds{W}$ is a Hermitian operator, it can be decomposed as $\ln\mds{A}+\mds{W}=\sum_\imath a_{\imath}\dyad{\imath}$, where $\{a_\imath\}_\imath$ are real eigenvalues.
Substituting this into the above expression, we obtain
\begin{align}
	\eval{\dv[2]{f(t)}{t}}_{t=0}&=\sum_{\imath,\jmath}\int_0^1\dd{s}|\mel{\imath}{\mds{W}'}{\jmath}|^2e^{sa_\jmath + (1-s)a_\imath}.
\end{align}
Since
\begin{equation}
	\int_0^1\dd{s}e^{sa_\jmath + (1-s)a_\imath}=\begin{cases}
		e^{a_\imath} & \text{if}~a_\imath=a_\jmath\\
		\frac{e^{a_\imath}-e^{a_\jmath}}{a_\imath-a_\jmath} & \text{if}~a_\imath\neq a_\jmath
	\end{cases},
\end{equation}
which is always positive, $\dv*[2]{f(0)}{t}\ge 0$.
The equality can be achieved only when $\mds{W}'=0$, or $\mds{U}'=c\mds{1}=-\mds{V}'$.
If we restrict to the subspace where $\tr\mds{V}=0$, then $c=0$, and $\mds{U}'=\mds{V}'=0$.
Therefore, $F$ is strictly convex on the restricted subspace.
Furthermore, since $F$ is continuous and coercive, it follows that there is a unique global minimum, which completes the proof.
\end{proof}

\subsection{Lower and upper bounds on the geometric complexity}
We prove the following lower and upper bounds:
\begin{equation}
	\ell\le\mca{C}(\Lambda)\le(\sqrt{d}+1)\ell,
\end{equation}
where $\ell=\|\ln d + \ln\Lambda(\mds{1}/d)\|_F$.

\subsubsection{Lower bound}
The lower bound follows directly from the properties of the defined Riemannian metric.
Let $\{\mds{M}_t\}_{0 \le t \le 1}$ be a smooth path of Choi matrices connecting $\mds{M}^{\Id}$ and $\mds{M}^\Lambda$.
Then $\phi_{\mds{M}_0}=\mds{1}/d$ and $\phi_{\mds{M}_1}=\Lambda(\mds{1}/d)$.
Using the result in Sec.~\ref{app:geo.dis.pos.Herm} that, on the unrestricted space of positive operators, the geodesic distance induced by $g_{\phi_{\mds{M}}}(\phi_{\mds{X}},\phi_{\mds{Y}})=\tr(\phi_{\mds{M}}^{-1}\phi_{\mds{X}}\phi_{\mds{M}}^{-1}\phi_{\mds{Y}})$ is $\mca{D}(\mds{A},\mds{B})=\|\ln(\mds{A}^{-1/2}\mds{B}\mds{A}^{-1/2})\|_F$, we obtain the following lower bound on the geometric complexity:
\begin{align}
	\mca{C}(\Lambda)&=\min_{\{\Lambda_t\}}\int_0^1\dd{t}\sqrt{g_{\mds{M}_t}(\dot{\mds{M}}_t,\dot{\mds{M}}_t)}\notag\\
	&=\min_{\{\phi_{\mds{M}_t}\}}\int_0^1\dd{t}\sqrt{g_{\phi_{\mds{M}_t}}(\dot{\phi}_{\mds{M}_t},\dot{\phi}_{\mds{M}_t})}\notag\\
	&\ge \mca{D}(\phi_{\mds{M}_0},\phi_{\mds{M}_1})\notag\\
	&=\|\ln(d) + \ln\Lambda(\mds{1}/d)\|_F\notag\\
	&=\ell.
\end{align}

\subsubsection{Upper bound}
We consider a specific path $\{\phi_t\}_{0\le t\le 1}$, defined by
\begin{align}
	\phi_t &= \frac{\phi_0^{1/2}(\phi_0^{-1/2}\phi_1\phi_0^{-1/2})^t\phi_0^{1/2}}{\tr[\phi_0^{1/2}(\phi_0^{-1/2}\phi_1\phi_0^{-1/2})^t\phi_0^{1/2}]},\notag\\
	\phi_0&=\mds{1}/d,~\text{and}~\phi_1=\Lambda(\mds{1}/d).
\end{align}
Evidently, $\phi_t$ is always a density matrix, and the path $\{\phi_t\}_{0\le t\le 1}$ smoothly connects the initial and final states $\phi_0$ and $\phi_1$.
For notational convenience, we define $\widetilde{\phi}_t\coloneqq\phi_0^{1/2}(\phi_0^{-1/2}\phi_1\phi_0^{-1/2})^t\phi_0^{1/2}$, $Z_t\coloneqq \tr \widetilde{\phi}_t$, and $A\coloneqq \phi_0^{-1/2}\phi_1\phi_0^{-1/2}$.

First, we evaluate the Riemannian metric $g_{\phi_t}(\dot\phi_t,\dot\phi_t)$.
To this end, we evaluate the following quantities:
\begin{align}
	\dot{\widetilde{\phi}}_t&=\phi_0^{1/2}A^t\ln A\phi_0^{1/2}=\widetilde{\phi}_tB,\\
	\dot Z_t&=\tr(\phi_0^{1/2}A^t\ln A\phi_0^{1/2})=\tr(\widetilde{\phi}_tB),
\end{align}
where $B\coloneqq\phi_0^{-1/2}\ln A\phi_0^{1/2}$.
Using these expressions, the time derivative of $\phi_t=\widetilde{\phi}_t/Z_t$ can be calculated as
\begin{align}
	\dot\phi_t&=\frac{\dot{\widetilde{\phi}}_t}{Z_t}-\phi_t\frac{\dot Z_t}{Z_t}\notag\\
	&=\frac{\widetilde{\phi}_t}{Z_t}B-\phi_t\tr(\frac{\widetilde{\phi}_t}{Z_t}B)\notag\\
	&=\phi_tB-\phi_t\tr(\phi_tB).
\end{align}
It follows that $\phi_t^{-1}\dot\phi_t=B-\tr(\phi_tB)$.
Consequently, the metric $g_{\phi_t}(\dot\phi_t,\dot\phi_t)$ can be written as
\begin{align}
	g_{\phi_t}(\dot\phi_t,\dot\phi_t)&=\tr{[B-\tr(\phi_tB)]^2}\notag\\
	&=\tr(B^2)-2\tr(B)\tr(\phi_tB)+d\tr(\phi_tB)^2.
\end{align}
Next, by applying the inequalities $|\tr(X^\dagger Y)|^2\le\tr(X^\dagger X)\tr(Y^\dagger Y)$ and $\tr(\phi_t^2)\le 1$, we have
\begin{align}
	g_{\phi_t}(\dot\phi_t,\dot\phi_t)&=\tr(B^2)-2\tr(B)\tr(\phi_tB)+d\tr(\phi_tB)^2\notag\\
	&\le (\sqrt{d}+1)^2\tr(B^2).
\end{align}
Since $\phi_0=\mds{1}/d$, $B=\ln A=\ln(d\phi_1)$. 
Using these facts, the geometric complexity is upper bounded as follows:
\begin{align}
	\mca{C}(\Lambda)&=\min_{\{\Lambda_t\}}\int_0^1\dd{t}\sqrt{g_{\mds{M}_t}(\dot{\mds{M}}_t,\dot{\mds{M}}_t)}\notag\\
	&=\min_{\{\phi_{\mds{M}_t}\}}\int_0^1\dd{t}\sqrt{g_{\phi_{\mds{M}_t}}(\dot{\phi}_{\mds{M}_t},\dot{\phi}_{\mds{M}_t})}\notag\\
	&\le \int_0^1\dd{t}\sqrt{g_{\phi_t}(\dot\phi_t,\dot\phi_t)}\notag\\
	&\le (\sqrt{d}+1)\sqrt{\tr{[\ln(d\phi_1)]^2}}\notag\\
	&=(\sqrt{d}+1)\|\ln d + \ln\phi_1\|_F\notag\\
	&=(\sqrt{d}+1)\ell.
\end{align}
Here, the minimum in the second line is taken over all continuous paths in the space of density operators connecting $\phi_0$ and $\phi_1$, due to the existence of a path of quantum channels proved in Theorem \ref{theo:qmap.exist}.

\subsection{Protocol-scaling relations for quantum maps}
\subsubsection{Lindbladian maps}
Let $\Lambda(\circ)=e^{\mca{L}_N}\dots e^{\mca{L}_1}(\circ)$ be a quantum channel realized by sequentially applying $N$ protocols, where each protocol is generated by Lindblad dynamics of the form:
\begin{equation}
	\mca{L}_k(\circ)\coloneqq -i[\mds{H}_k,\circ] + \sum_{c}\qty[\mds{L}_{k,c}\circ \mds{L}_{k,c}^\dagger - \frac{1}{2}\{\mds{L}_{k,c}^\dagger \mds{L}_{k,c},\circ\}].
\end{equation}
For simplicity, we assume that each protocol is applied over a unit time interval.
As a physically reasonable constraint on the coupling strength between the system and the environment, we impose that $\|\sum_c\mds{L}_{k,c}^\dagger \mds{L}_{k,c}\|_\infty\le\gamma$ for any $k$.
Let $\phi_t$ be the density operator at time $t$ obtained by applying these protocols to the maximally mixed state, i.e., $\phi_t=\mca{T}e^{\int_0^t\dd{s}\mca{L}_t}(\mds{1}/d)$, where $\mca{L}_t$ is one of the superoperators $\{\mca{L}_1,\dots,\mca{L}_N\}$ applied at time $t$. Let $\phi_t=\sum_{n}p_n(t)\dyad{n_t}$ be the spectral decomposition of $\phi_t$. Then the time evolution of $p_n(t)$ is governed by
\begin{equation}
	\dot p_n(t)=\sum_{c,m}[w_{nm}^{k,c}(t)p_m(t) - w_{mn}^{k,c}(t)p_n(t)],
\end{equation}
where $w_{mn}^{k,c}(t)\coloneqq|\mel{m_t}{\mds{L}_{k,c}}{n_t}|^2$, and the index $k$ depends on time $t$.
Applying H{\"o}lder's inequality, $|\tr(\mds{A}^\dagger \mds{B})|\le\|\mds{A}\|_1\|\mds{B}\|_\infty$, we obtain
\begin{align}
	-\frac{\dot p_n(t)}{p_n(t)}&=\sum_{c,m}w_{mn}^{k,c}(t) - \frac{1}{p_n(t)}\sum_{c,m}w_{nm}^{k,c}(t)p_m(t)\notag\\
	&\le \sum_{c,m}w_{mn}^{k,c}(t)\notag\\
	&=\mel{n_t}{\sum_c\mds{L}_{k,c}^\dagger \mds{L}_{k,c}}{n_t}\notag\\
	&\le \|\sum_{c}\mds{L}_{k,c}^\dagger \mds{L}_{k,c}\|_\infty\notag\\
	&\le\gamma.
\end{align}
Integrating this inequality over time, we obtain
\begin{equation}
	\ln(1/d) - \ln p_n(N) \le \int_0^N\dd{t}\gamma = N\gamma.
\end{equation}
On the other hand, since $p_n(N)\in[0,1]$, we have
\begin{equation}
	\ln p_n(N) - \ln(1/d) \le \ln d.
\end{equation}
Therefore,
\begin{equation}
	|\ln p_n(N) - \ln(1/d)| \le \max(N\gamma,\ln d) \le N\ln(de^\gamma).
\end{equation}
Since $\phi_N=\Lambda(\mds{1}/d)$, the Frobenius norm is upper bounded as
\begin{align}
	\|\ln d + \ln\Lambda(\mds{1}/d)\|_F&=\sqrt{\tr[(\ln d + \ln\phi_N)^2]}\notag\\
	&=\sqrt{\sum_n[\ln d + \ln p_n(N)]^2}.
\end{align}
Consequently,
\begin{equation}
	\ell=\|\ln d + \ln\Lambda(\mds{1}/d)\|_F\le N\sqrt{d}\ln(de^\gamma).
\end{equation}
Combining this inequality with the fact $\mca{C}(\Lambda)\le (\sqrt{d}+1)\ell$ yields the following lower bound on the number of protocols:
\begin{equation}
	N\ge \frac{\mca{C}(\Lambda)}{(d+\sqrt{d})\ln(de^\gamma)}.
\end{equation}

\subsubsection{General unitary dynamics}

We consider the case $\Lambda=\Lambda_N\dots\Lambda_1$, where each map $\Lambda_k$ is implemented through general unitary dynamics of the system and an environment $E_k$.
That is, $\Lambda_k[\circ]=\tr_{E_k}[\mds{U}_k(\circ\otimes\pi_k)\mds{U}_k^\dagger]$, where $\pi_k=e^{-\beta \mds{H}_k}/Z$ is the thermal Gibbs state of the environment $E_k$, and $\mds{U}_k$ is the unitary operator acting on the joint system composed of the system and the environment.
Under physical constraints, it is natural to assume the following limit on the control bandwidth: $\beta\Delta\mds{H}_k\le\gamma$.
Here, $\Delta\mds{H}\coloneqq\max_n E_n-\min_n E_n$, where $\{E_n\}_n$ are the eigenvalues of the Hermitian operator $\mds{H}$.

Defining $\phi_0\coloneqq\mds{1}/d$ and $\phi_k\coloneqq\Lambda_k\dots\Lambda_1(\phi_0)$, we obtain
\begin{align}
	\mca{C}(\Lambda)
	&\le(\sqrt{d}+1)\mca{D}(\phi_0,\phi_N)\notag\\
	&=(\sqrt{d}+1)\|\ln(d\phi_N)\|_F\notag\\
	&=(\sqrt{d}+1)\sqrt{\sum_n(\ln[d\lambda_n(\phi_N)])^2},
\end{align}
where $0\le\lambda_1(\phi)\le\dots\le\lambda_d(\phi)\le 1$ are the eigenvalues of the density operator $\phi$ in increasing order.
Since $|\ln[d\lambda_n(\phi)]|\le\max\!\qty(|\ln[d\lambda_1(\phi)]|,\ln d)$, we readily obtain
\begin{align}
	\mca{C}(\Lambda)
	&\le \sqrt{d}(\sqrt{d}+1)\max\!\qty(|\ln[d\lambda_1(\phi_N)]|,\ln d)\notag\\
	&=(d+\sqrt{d})\max\!\qty(\ln\qty[\frac{1}{d\lambda_1(\phi_N)}],\ln d),
\end{align}
where we used the fact that $d\lambda_1(\phi)\le 1$.
Next, we lower-bound the eigenvalue $\lambda_1(\phi_N)$.
From the relation $\phi_k=\Lambda_k(\phi_{k-1})=\tr_{E_k}[\mds{U}_k(\phi_{k-1}\otimes\pi_k)\mds{U}_k^\dagger]$, we have
\begin{align}
	\lambda_1(\phi_k)\ge \lambda_1(\phi_{k-1})\lambda_1(\pi_k).
\end{align}
Therefore, $\lambda_1(\phi_N)\ge\lambda_1(\phi_0)\prod_{k=1}^N\lambda_1(\pi_k)$.
Since $\lambda_1(\pi_k)\ge e^{-\beta\Delta\mds{H}_k}/d\ge e^{-\gamma}/d$, we obtain the following lower bound on $\lambda_1(\phi_N)$:
\begin{align}
	\lambda_1(\phi_N)\ge \frac{1}{d}\qty(\frac{e^{-\gamma}}{d})^N=\frac{1}{d}\qty(\frac{1}{de^\gamma})^N,
\end{align}
which immediately yields
\begin{align}
	\ln\qty[\frac{1}{d\lambda_1(\phi_N)}]\le N\ln(de^{\gamma}).
\end{align}
Note that $N\ln(de^{\gamma})\ge \ln d$.
Consequently, the geometric complexity is upper-bounded as
\begin{align}
	\mca{C}(\Lambda)
	&\le N(d+\sqrt{d})\ln(de^{\gamma}).
\end{align}
This yields the following lower bound on the number of protocols in terms of geometric complexity:
\begin{equation}
	N\ge\frac{\mca{C}(\Lambda)}{(d+\sqrt{d})\ln(de^\gamma)}.
\end{equation}

\subsection{Proof of the entropic bound on the geometric complexity}

We prove the following lower bound on the geometric complexity:
\begin{equation}
	\mca{C}(\Lambda)\ge\ln S(\mds{1}/d) - \ln S(\Lambda[\mds{1}/d]).
\end{equation}
Here, $S(\rho)=-\tr(\rho\ln\rho)$ is the von Neumann entropy of the quantum state $\rho$.
Let $\{\Lambda_t\}_{0\le t\le 1}$ be the geodesic path, and define $\phi_t\coloneqq\phi_{\mds{M}_t}=\Lambda_t(\mds{1}/d)$, where $\mds{M}_t$ denotes the Choi matrix of the quantum channel $\Lambda_t$.
We first show that
\begin{equation}\label{eq:meb.tmp1}
	g_{\phi_t}(\dot\phi_t,\dot\phi_t)\ge \qty|\dv{t}\ln S(\phi_t)|^2.
\end{equation}
Using the spectral decomposition $\phi_t=\sum_n p_n(t)\dyad{n_t}$, the right-hand side is upper bounded as
\begin{align}
	\qty|\dv{t}\ln S(\phi_t)|^2&=\qty|\frac{\tr(\dot\phi_t\ln\phi_t)}{\tr(\phi_t\ln\phi_t)} |^2\notag\\
	&=\qty|\frac{\sum_n\mel{n_t}{\dot\phi_t}{n_t}\ln p_n(t)}{\sum_np_n(t)\ln p_n(t)}|^2\notag\\
	&\le\frac{[\sum_n\mel{n_t}{\dot\phi_t}{n_t}\ln p_n(t)]^2}{\sum_n[p_n(t)\ln p_n(t)]^2}\notag\\
	&\le\sum_n\qty[\frac{\mel{n_t}{\dot\phi_t}{n_t}\ln p_n(t)}{p_n(t)\ln p_n(t)}]^2\notag\\
	&=\sum_n\qty[\frac{\mel{n_t}{\dot\phi_t}{n_t}}{p_n(t)}]^2.\label{eq:meb.tmp2}
\end{align}
Here, we use $(\sum_n p_n\ln p_n)^2\ge\sum_n(p_n\ln p_n)^2$ for any distribution $\{p_n\}$ to obtain the third line, and $\sum_n a_n^2/b_n\ge(\sum_n a_n)^2/\sum_n b_n$ for nonnegative $\{b_n\}$ and real $\{a_n\}$ to obtain the fourth line.
On the other hand, the metric can be lower bounded as
\begin{align}
	g_{\phi_t}(\dot\phi_t,\dot\phi_t)&=\tr(\phi_t^{-1}\dot\phi_t\phi_t^{-1}\dot\phi_t)\notag\\
	&=\sum_{m,n}\frac{|\mel{m_t}{\dot\phi_t}{n_t}|^2}{p_m(t)p_n(t)}\notag\\
	&\ge\sum_n\qty[\frac{\mel{n_t}{\dot\phi_t}{n_t}}{p_n(t)}]^2.\label{eq:meb.tmp3}
\end{align}
Combining Eqs.~\eqref{eq:meb.tmp2} and \eqref{eq:meb.tmp3}, we immediately obtain Eq.~\eqref{eq:meb.tmp1}.
Consequently, the geometric complexity is bounded as follows:
\begin{align}
	\mca{C}(\Lambda)&=\int_0^1\dd{t}\sqrt{g_{\phi_t}(\dot{\phi}_t,\dot{\phi}_t)}\notag\\
	&\ge \int_0^1\dd{t}\qty|\dv{t}\ln S(\phi_t)|\notag\\
	&\ge \qty|\int_0^1\dd{t}\dv{t}\ln S(\phi_t)|\notag\\
	&=|\ln S(\phi_0)-\ln S(\phi_1) |\notag\\
	&=\ln S(\mds{1}/d) - \ln S(\Lambda[\mds{1}/d]).
\end{align}

\subsection{Proof of the trade-off relation \eqref{eq:qua.comp.3rd.law}}
Let $\phi=\Lambda(\mds{1}/d)$, and let $0\le\lambda_1\le\dots\le\lambda_d\le 1$ be its eigenvalues.
Since $\sum_n\lambda_n=1$, we have $d\lambda_1\le 1$.
Using this fact and the inequality $\sqrt{\sum_n x_n^2}\ge\max_n|x_n|$, $\ell$ can be lower bounded as
\begin{align}
	\ell&=\|\ln(d\phi)\|_F=\sqrt{\sum_{n=1}^d\ln(d\lambda_n)^2}\notag\\
	&\ge |\ln(d\lambda_1)|=-\ln(d\lambda_1).
\end{align}
Therefore, $e^{\mca{C}(\Lambda)}\ge e^{\ell}\ge 1/(d\lambda_1)$.
On the other hand,
\begin{align}
	\epsilon(\Lambda)=d\tr(\Pi\phi)\ge d\lambda_1.
\end{align}
Consequently, multiplying these two inequalities yields the desired trade-off relation:
\begin{equation}
	e^{\mca{C}(\Lambda)}\times \epsilon(\Lambda)\ge 1.
\end{equation}

\subsection{Calculation of the geometric complexity for a quantum reset channel}
We derive an explicit expression for the geometric complexity of the following quantum channel:
\begin{equation}
	\Lambda(\circ)=\tr_E[\mds{U}(\circ\otimes\rho_E)\mds{U}^\dagger],
\end{equation}
where the unitary operator $\mds{U}$ swaps the states of the system and the environment, and $\rho_E=\kappa\dyad{1}+(1-\kappa)\dyad{0}$ is a quantum state close to the ground state.
In particular,
\begin{equation}\label{eq:swap.uni}
	\mds{U}\qty(\mds{1}/2\otimes\rho_E)\mds{U}^\dagger=\rho_E\otimes\mds{1}/2.
\end{equation}
Thus, for arbitrary initial states, the quantum channel $\Lambda$ reliably resets the system close to the ground state, with an error bounded from above by $2\kappa$ for the projector $\Pi=\dyad{1}$.

Given the spectral decomposition $\rho_E=\sum_n p_n\dyad{n}$, Eq.~\eqref{eq:swap.uni} is equivalent to
\begin{equation}
	\sum_{m,n}p_n\mds{U}(\dyad{m}\otimes\dyad{n})\mds{U}^\dagger=\sum_{m,n}p_n\dyad{n}\otimes\dyad{m}.
\end{equation}
Hence, the unitary operator is given by $\mds{U}=\sum_{n,m}\dyad{m,n}{n,m}$, where we use $\ket{m,n}=\ket{m}\otimes\ket{n}$ for notational convenience.
The quantum channel $\Lambda$ can then be written as
\begin{align}
	\Lambda(\circ)&=\sum_{m,n}p_m\dyad{m}{n}\circ\dyad{n}{m}\notag\\
	&=\sum_{m,n}\mds{K}_{mn}\circ\mds{K}_{mn}^\dagger,
\end{align}
where $\mds{K}_{mn}=\sqrt{p_m}\dyad{m}{n}$.
Since $\Lambda(\mds{1}/2)=\rho_E$, the geometric complexity of this quantum channel reads
\begin{equation}
	\mca{C}(\Lambda)=\min_{\{\phi_t\}}\int_0^1\dd{t}\sqrt{g_{\phi_t}(\dot\phi_t,\dot\phi_t)},
\end{equation}
where the minimum is taken over all smooth paths of quantum states connecting $\mds{1}/2$ and $\rho_E$.

Using the Bloch representation, each quantum state $\phi_t$ can be written as
\begin{equation}
	\phi_t=\frac{1}{2}\qty(\mds{1}+\vb{a}_t\cdot\mbm{\sigma}),
\end{equation}
where $\mbm{\sigma}=(\sigma_x,\sigma_y,\sigma_z)$ are the Pauli matrices, $\vb{a}_t$ is the Bloch vector satisfying $\|\vb{a}_t\|_2\le 1$, and $\|\cdot\|_2$ denotes the Euclidean norm.
The inverse matrix $\phi_t^{-1}$ can also be expressed in terms of $\vb{a}_t$ and $\mbm{\sigma}$ as
\begin{equation}
	\phi_t^{-1}=\frac{2}{1-\|\vb{a}_t\|_2^2}(\mds{1}-\vb{a}_t\cdot\mbm{\sigma}).
\end{equation}
Using $\dot\phi_t=\dot{\vb{a}}_t\cdot\mbm{\sigma}/2$ and $(\vb{a}\cdot\mbm{\sigma})(\vb{b}\cdot\mbm{\sigma})=(\vb{a}\cdot\vb{b})\mds{1}+i(\vb{a}\times\vb{b})\cdot\mbm{\sigma}$, we obtain
\begin{align}
	\phi_t^{-1}\dot{\phi}_t&=\frac{1}{1-\|\vb{a}_t\|_2^2}(\mds{1}-\vb{a}_t\cdot\mbm{\sigma})(\dot{\vb{a}}_t\cdot\mbm{\sigma})\notag\\
	&=\frac{1}{1-\|\vb{a}_t\|_2^2}\qty[(\dot{\vb{a}}_t\cdot\mbm{\sigma}) - (\vb{a}_t\cdot\dot{\vb{a}}_t)\mds{1} - i(\vb{a}_t\times\dot{\vb{a}}_t)\cdot\mbm{\sigma}]\notag\\
	&=\frac{1}{1-\|\vb{a}_t\|_2^2}\qty[ - (\vb{a}_t\cdot\dot{\vb{a}}_t)\mds{1} + (\dot{\vb{a}}_t - i(\vb{a}_t\times\dot{\vb{a}}_t))\cdot\mbm{\sigma} ].
\end{align}
Let $\alpha=- (\vb{a}_t\cdot\dot{\vb{a}}_t)$ and $\vb{b}_t=\dot{\vb{a}}_t - i(\vb{a}_t\times\dot{\vb{a}}_t)$. Then, the metric is calculated as
\begin{align}
	g_{\phi_t}(\dot\phi_t,\dot\phi_t)&=\frac{1}{(1-\|\vb{a}_t\|_2^2)^2}\tr[(\alpha\mds{1}+\vb{b}_t\cdot\mbm{\sigma})^2]\notag\\
	&=\frac{1}{(1-\|\vb{a}_t\|_2^2)^2}\tr[\alpha^2\mds{1} + 2\alpha(\vb{b}_t\cdot\mbm{\sigma}) +  (\vb{b}_t\cdot\mbm{\sigma})^2]\notag\\
	&=\frac{1}{(1-\|\vb{a}_t\|_2^2)^2}\tr[\alpha^2\mds{1} + 2\alpha(\vb{b}_t\cdot\mbm{\sigma}) +  \|\vb{b}_t\|_2^2\mds{1}]\notag\\
	&=\frac{2}{(1-\|\vb{a}_t\|_2^2)^2}\qty(\alpha^2+\|\vb{b}_t\|_2^2).
\end{align}
Since
\begin{align}
	\|\vb{b}_t\|_2^2&=[\dot{\vb{a}}_t - i(\vb{a}_t\times\dot{\vb{a}}_t)]\cdot[\dot{\vb{a}}_t - i(\vb{a}_t\times\dot{\vb{a}}_t)]\notag\\
	&=\|\dot{\vb{a}}_t\|_2^2 - \|\vb{a}_t\times\dot{\vb{a}}_t\|_2^2\notag\\
	&=\|\dot{\vb{a}}_t\|_2^2 - [\|\vb{a}_t\|_2^2\|\dot{\vb{a}}_t\|_2^2 - (\vb{a}_t\cdot\dot{\vb{a}}_t)^2 ]\notag\\
	&=\|\dot{\vb{a}}_t\|_2^2(1-\|\vb{a}_t\|_2^2) + (\vb{a}_t\cdot\dot{\vb{a}}_t)^2,
\end{align}
we obtain
\begin{equation}
	g_{\phi_t}(\dot\phi_t,\dot\phi_t)=\frac{2\|\dot{\vb{a}}_t\|_2^2}{1-\|\vb{a}_t\|_2^2}+\frac{4(\vb{a}_t\cdot\dot{\vb{a}}_t)^2}{(1-\|\vb{a}_t\|_2^2)^2}.
\end{equation}
To proceed, it is convenient to use spherical coordinates,
\begin{equation}
	\vb{a}_t=[a_t\sin\theta_t\cos\varphi_t,a_t\sin\theta_t\sin\varphi_t,a_t\cos\theta_t]^\top.
\end{equation}
Then $\|\dot{\vb{a}}_t\|_2^2=\dot a_t^2 + a_t^2\dot\theta_t^2+a_t^2(\sin\theta_t)^2\dot\varphi_t^2$ and $\vb{a}_t\cdot\dot{\vb{a}}_t=a_t\dot a_t$.
Therefore,
\begin{align}
	g_{\phi_t}(\dot\phi_t,\dot\phi_t)&=\frac{2(\dot a_t^2 + a_t^2\dot\theta_t^2+a_t^2(\sin\theta_t)^2\dot\varphi_t^2)}{1-a_t^2}+\frac{4a_t^2\dot a_t^2}{(1-a_t^2)^2}\notag\\
	&=\frac{2(1+a_t^2)}{(1-a_t^2)^2}\dot a_t^2 + \frac{2a_t^2}{1-a_t^2}\dot\theta_t^2 + \frac{2a_t^2(\sin\theta_t)^2}{1-a_t^2}\dot\varphi_t^2.
\end{align}
Since $a_t\in[0,1]$, the geodesic connecting the center of the sphere to any point lies in a fixed plane through the origin, where both $\theta_t$ and $\varphi_t$ are constant.
Therefore, the metric simplifies to
\begin{equation}
	g_{\phi_t}(\dot\phi_t,\dot\phi_t)=\frac{2(1+a_t^2)}{(1-a_t^2)^2}\dot a_t^2.
\end{equation}
Consequently, the geometric complexity can be evaluated as
\begin{align}
	\mca{C}(\Lambda)&=\int_{0}^{a_1}\dd{a}\frac{\sqrt{2(1+a^2)}}{1-a^2}\notag\\
	&=2\atanh\frac{\sqrt{2}a_1}{\sqrt{1+a_1^2}} - \sqrt{2}\atanh\frac{a_1}{\sqrt{1+a_1^2}}.
\end{align}
Defining $\wp\coloneqq\sqrt{2\tr(\rho_E^2)-1}=|2\kappa-1|$, we can express the geometric complexity as
\begin{equation}
	\mca{C}(\Lambda)=2\atanh\frac{\sqrt{2}\wp}{\sqrt{1+\wp^2}} - \sqrt{2}\atanh\frac{\wp}{\sqrt{1+\wp^2}}.
\end{equation}
As a result, $\mca{C}(\Lambda)$ diverges in the limits $\kappa\to 0$ and $\kappa\to 1$, that is, when the initial state of the environment becomes pure.
Thus, the geometric complexity captures the purity of the environment, which is one of the crucial resources for cooling.

\end{document}